\documentclass[journal,onecolumn,draftcls]{IEEEtran}

\newtheorem{thm}{Theorem}
\newtheorem{cor}{Corollary}

\newtheorem{prop}{Proposition}
\newtheorem{defn}{Definition}
\newtheorem{rem}{Remark}

\usepackage[final]{graphicx}
\usepackage[reqno]{amsmath}
\usepackage{amssymb}
\usepackage{cite}
\usepackage{epstopdf}

\newcommand{\cv}{{\bf c}}

\newcommand{\qv}{{\bf q}}

\newcommand{\sv}{{\bf s}}

\newcommand{\uv}{{\bf u}}

\newcommand{\xv}{{\bf x}}
\newcommand{\yv}{{\bf y}}

\newcommand{\Um}{{\bf U}}

\newcommand{\Zm}{{\bf Z}}

\newcommand{\Cc}{{\cal C}}

\newcommand{\Pc}{{\cal P}}
\newcommand{\Qc}{{\cal Q}}
\newcommand{\Rc}{{\cal R}}

\newcommand{\Xc}{{\cal X}}
\newcommand{\Yc}{{\cal Y}}

\begin{document}

\title{The Two-Way Wiretap Channel: Achievable Regions and Experimental Results}

\author{\IEEEauthorblockN{Aly El Gamal, O. Ozan Koyluoglu,
Moustafa Youssef, and Hesham El Gamal}
\thanks{Aly El Gamal was with the Wireless Intelligent Networks Center (WINC), Nile University, Cairo, Egypt.
He is now with the University of Illinois at Urbana-Champaign
(Email: elgamal1@uiuc.edu).
O. Ozan Koyluoglu was with the Department
of Electrical and Computer Engineering, The Ohio State University, Columbus, OH.
He is now with the University of Texas at Austin (Email: ozan@austin.utexas.edu).
Hesham El Gamal is with the Department
of Electrical and Computer Engineering, The Ohio State University,
Columbus, OH (Email: helgamal@ece.osu.edu).
Moustafa Youssef was with the Wireless Intelligent
Networks Center (WINC), Nile University, Cairo, Egypt.
He is now with Alexandria University and Egypt-Japan University of Science and Technology (E-JUST)
(Email: moustafa.youssef@ejust.edu.eg).}
\thanks{This work was presented in part at the 2009 IEEE Global
Communications Conference (GLOBECOM 2009)
and the 2010 IEEE Information Theory Workshop (ITW 2010).}
\thanks{This research was supported in part by the National
Science Foundation (NSF), the Los Alamos National Labs (LANL), the USAID Fund, and QNRF.}
}

\maketitle

\begin{abstract}
This work considers the two-way wiretap
channel in which two legitimate users, Alice and Bob, wish to exchange messages securely
in the presence of a passive eavesdropper Eve. In the full-duplex scenario, where each node can transmit and receive simultaneously, we obtain new achievable secrecy rate regions based on the idea of allowing the two users to \emph{jointly} optimize their channel prefixing
distributions and binning codebooks in addition to key sharing. The new regions are shown to be strictly larger than the known ones for a wide class of discrete memoryless and Gaussian channels. In the half-duplex case, where a user can only transmit or receive on any given degree of freedom, we introduce the idea of {\em randomized scheduling} and establish the significant gain it offers in terms of the achievable secrecy sum-rate. We further develop an experimental setup based on a IEEE 802.15.4-enabled sensor boards, and use this testbed to show that one can exploit the two-way nature of the communication, via appropriately randomizing the transmit power levels and transmission schedule, to introduce significant ambiguity at {\bf a noiseless} Eve. 
\end{abstract}

\section{Introduction}
In a pioneering paper~\cite{Shannon:BSTJ:49}, Shannon
established the achievability of perfectly secure communication
in the presence of an eavesdropper with unbounded computational
complexity. However, the necessary condition for perfect secrecy,
i.e., that the entropy of the private key is at least as large as that of
the message, appears to be prohibitive for most practical applications.
In~\cite{wyner}, Wyner revisited the problem and proved the
achievability of a positive secrecy rate over a degraded discrete
memoryless channel, via a {\em key-less} secrecy approach, by
relaxing the \emph{noiseless} assumption and the strict notion
of perfect secrecy employed in~\cite{Shannon:BSTJ:49}. Wyner's
results were later extended to the Gaussian and broadcast channels
in~\cite{GWT}~and~\cite{BCC}, respectively. In~\cite{public-discussion},
Maurer showed how to exploit the presence of a \emph{public discussion}
channel to achieve positive secrecy over the one way wiretap channel
even when the eavesdropper channel is less noisy than the legitimate one.
In~\cite{Lai2008}, the authors considered a more practical
feedback scenario where the noiseless public
channel is replaced by {\em receiver feedback} over
the same noisy channel. Under this assumption, it was shown
that the perfect secrecy capacity is equal to the capacity of
the main channel in the absence of the eavesdropper for
full-duplex modulo-additive discrete memoryless channels.
More interestingly,~\cite{Lai2008} established the achievability
of positive secrecy rates, even under the half-duplex constraint
where each feedback symbol introduces an erasure event in the main channel.

Our work generalizes this line of work by investigating the fundamental limits of the two-way wiretap channel, where Alice and Bob wish
to exchange secure messages in the presence of a passive eavesdropper Eve.
It is easy to see that the one way channel with feedback considered
in~\cite{Lai2008} is a special case of this model. Using the cooperative
channel prefixing and binning technique proposed in~\cite{Koyluoglu:Onthe08,Koyluoglu:Cooperative}, along with an innovative approach for key sharing between Alice and Bob,
we first derive an inner bound on the secrecy capacity region of the full-duplex discrete
memoryless two-way wiretap channel. By specializing our results to the additive modulo-$2$ and Gaussian channel, our region is shown to be strictly larger than those reported
recently in the literature~\cite{TW-E,GGMAC-WT,HeAndYener}. The gain can be attributed to the fact that we allow both nodes to simultaneously send secure messages when the channel conditions are favorable. We then proceed to the half-duplex setting where each node can only transmit or receive on the same degree of freedom. Here, we introduce the concept of {\em randomized scheduling for secrecy}, whereby Alice and Bob send their symbols
at random time instants to maximally confuse Eve at the expense of introducing {\bf collisions and erasure events} in the main channel. Remarkably, this approach is shown to result in significant gains in the achievable secure sum rate, as compared with the traditional deterministic scheduling approach. In the Gaussian scenario, we show that the ambiguity at Eve can be further enhanced by randomizing the transmit power levels.

Inspired by our information theoretic foundation, we develop an IEEE $802.15.4$ testbed to estimate the ambiguity at the eavesdropper in near field wireless sensor networks where the distance between the legitimate nodes is significantly smaller than that to the potential eavesdropper. A representative scenario corresponds to Body Area Networks (BAN) which are being considered for a variety of health care applications. Here, the sensor nodes are mounted on the body, and hence, any potential eavesdropper is expected to be at a significantly larger distance from each legitimate node. Clearly, ensuring the confidentiality of the messages exchanged between sensors is an important design consideration in this application. Assuming an eavesdropper equipped with an energy classifier, analytical and experimental results that quantify the achievable secrecy sum rate under a two dimensional path loss model are derived. However, it is worth noting that we do not address the issue of implementing the classical wiretap code~\cite{wyner} in this work. Overall, these results establish the gain offered by the two-way randomization concept and establish the feasibility of our approach in realistic scenarios.

It is worth noting that similar settings to the one considered in this work, exist in the literature. In particular, the authors in~\cite{Morsy-ICC11} consider a binary erasure block-fading channel where the nodes are placed according to a similar geometric model to that in Section~\ref{experimental}, and provide analytical and experimental results for the secrecy outage probabilities for frames of different sizes,~\cite{HeAndYener2} considers an extension of the two-way wiretap channel where the untrusted eavesdropper may be used to relay messages between the two users. Also,~\cite{Pierrot-Bloch} considers the two-way wiretap channel with a {\em strong secrecy} constraint, where the mutual information leakage to the eavesdropper, rather than the leakage rate (defined in Section~\ref{full_duplex}) is required to vanish in the limit of the number of channel uses.

The rest of the paper is organized as follows. In Section~\ref{full_duplex},
we develop an achievable secrecy rate region for the full-duplex discrete
memoryless two-way wiretap channel, and specialize the result to the
additive modulo-$2$ and Gaussian channel. Section~\ref{half_duplex} is devoted to the half-duplex scenario where the concept of randomized scheduling is introduced. Our practical setting, using the TinyOS-enabled sensor boards, is described in Section~\ref{experimental}. The analytical and experimental results of this section establish the feasibility of our approach in near field wireless sensor network applications. Finally,
we offer some concluding remarks in Section~\ref{conclusion}. To enhance the flow of the paper, the detailed proofs are collected in the appendices.

\section{Full-Duplex Channels}\label{full_duplex}
In the full-duplex scenario, each of the two legitimate terminals
is equipped with a transmitter and a receiver that can operate simultaneously on the same degree of freedom. The two users intend
to \emph{exchange} messages in the presence of a (passive)
eavesdropper. More specifically, the $i^{th}$ user wishes to transmit a secret message $w_i$,
selected from a set of \emph{equiprobable} messages ${\cal M}_i=\{1,\ldots,M_i\}$, to the other
user, in $n$ channel uses, where $i=1,2$. For message $w_i$, a codeword
${\bf X}_i(w_i)=\{X_i(1),\ldots,X_i(n)\}$ is transmitted at a rate
$R_i=\frac{1}{n}\log_2 M_i$. The $i^{th}$ decoder employs a
decoding function $\phi_{i}(.)$ to map the received sequence
${\bf Y}_{i}$ to an estimate $\hat{w}_i$ of $w_i$.
The two-way communication is governed by
\emph{reliability} and \emph{secrecy} constraints. The former is
measured by the average probability of error,
\begin{equation*}
P_{e,i}=\frac{1}{M_i}\sum\limits_{w_i \in {\cal M}_i}
P\{\hat{w}_i\neq w_i|w_i \textrm{ is sent}\},
\textrm{ for } i=1,2;
\end{equation*}
whereas the latter is quantified by the mutual information leakage rate
to the eavesdropper $L$, i.e.,
\begin{equation*}
L_n=\frac{1}{n}I(W_1,W_2;{\bf Z}),
\end{equation*}
where ${\bf Z}=\{Z(1),\ldots,Z(n)\}$ is the observed sequence
at the eavesdropper. Here, we focus on the {\em perfect secrecy}
rate region, where the leakage rate is made arbitrarily
small~\cite{wyner}, as formalized in the following.

\begin{defn}\label{eqn:achievable}
The secret rate tuple $(R_1,R_2)$ is achievable for the
two-way wiretap channel, if for any given $\epsilon>0$,
there exists an $(n,M_1,M_2,P_{e,1},P_{e,2},L_n)$ code such that,
\begin{eqnarray*}
R_1&=&\frac{1}{n}\log_2 M_1\\
R_2&=&\frac{1}{n}\log_2 M_2\\
\max(P_{e,1},P_{e,2})&\leq& \epsilon\\
L_n &\leq& \epsilon,
\end{eqnarray*}
for sufficiently large $n$.
\end{defn}
 
We note that the last condition implies that (see, e.g.,~\cite[Lemma 15]{Koyluoglu:Cooperative})
\begin{eqnarray*}
\frac{1}{n}H(W_i|{\bf Z})\geq R_i-\epsilon \textrm{ for } i=1,2.
\end{eqnarray*}

The secrecy capacity region is defined as the set of all achievable
secret rate tuples $(R_1,R_2)$ and is denoted by $\Cc^F$.
Throughout the sequel, we will use the following shorthand notation for probability
distributions: $P(x)\triangleq P(X=x)$, $P(x|y)\triangleq P(X=x|Y=y)$,
and $P(x,y)\triangleq P(X=x,Y=y)$, where $X$ and $Y$ denote arbitrary
random variables. We will also use $\log(x)$ to denote $\log_2(x)$, and
$[a]^{+}$ to denote $\max(a,0)$. Furthermore, for the full-duplex \emph{discrete memoryless two-way channel with an
external passive eavesdropper} (DM-TWC-E), we will use the calligraphic
letters ${\cal X}_1$ and ${\cal X}_2$ to denote the discrete input
finite alphabets for user $1$ and user $2$, respectively, and ${\cal Y}_1$,
${\cal Y}_2$, and ${\cal Z}$, to denote the output alphabets observed
at the decoders of user $1$, user $2$, and the eavesdropper, respectively.
The channel is given by $P(y_1,y_2,z|x_1,x_2)$ and is memoryless
in the following sense.
$$P(y_1(t),y_2(t),z(t)|{\bf x}_1^t,{\bf x}_2^t,
{\bf y}_1^{t-1},{\bf y}_2^{t-1},{\bf z}^{t-1})
= P(y_1(t),y_2(t),z(t)|x_1(t),x_2(t)).$$

We further assume all channel state information to be available at all nodes. Our general achievable region is obtained via a coding scheme inspired 
by~\cite{Koyluoglu:Cooperative} where the codewords
${\bf C}_1$ and ${\bf C}_2$ are drawn from the two binning codebooks, and
passed on to the two respective prefix channels. To maximize the
ambiguity at Eve, both the binning
codebooks and channel prefixing distributions are jointly optimized.
In addition, the proposed scheme involves key sharing with a
block encoding technique to facilitate the secrecy generation.
In particular, the key received from the other user during the previous
block is used in a one time pad scheme~\cite{vernam} to transmit
additional secret bits. The codeword consisting of the XOR
of the message and the key serves a) as a cloud center in the
superposition coding and b) as an additional randomization
for the binning codebook.
The following result characterizes the
set of achievable rates using our coding scheme.

\begin{thm}\label{thm:FullDuplexDMC-R}
The proposed coding scheme achieves the region $\Rc$ for
the full-duplex DM-TWC-E.
\begin{equation*}
{\cal R} \triangleq \textrm{ closure of }
\left\{\underset{p\in{\cal P}}{\bigcup} {\cal R}(p)\right\} \subseteq \Cc^F,
\end{equation*}
where ${\cal P}$ denotes the set of all joint distributions
of the random variables $Q$, $U_1$, $U_2$, $C_1$, $C_2$,
$X_1$, and $X_2$ satisfying
\begin{equation*}
P(q,u_1,u_2,c_1,c_2,x_1,x_2)=P(q)P(u_1|q)P(c_1|u_1)P(x_1|c_1)P(u_2|q)P(c_2|u_2)P(x_2|c_2)
\end{equation*}
and ${\cal R}(p)$ is the closure of all
rate pairs $(R_1=R_1^u+R_1^s+R_1^o,R_2=R_2^u+R_2^s+R_2^o)$,
with non-negative tuples
$(R_1^u,R_1^s,R_1^o,R_1^x,R_2^u,R_2^s,R_2^o,R_2^x)$
satisfying
\begin{eqnarray}
R_1^s+R_1^k+R_1^o+R_1^x &\leq& I(C_1;Y_2|X_2,U_1,Q) \label{eq:a1}\\
R_1^u+R_1^s+R_1^k+R_1^o+R_1^x &\leq& I(U_1,C_1;Y_2|X_2,Q) \label{eq:a2}\\
R_2^s+R_2^k+R_2^o+R_2^x &\leq& I(C_2;Y_1|X_1,U_2,Q) \label{eq:a3}\\
R_2^u+R_2^s+R_2^k+R_2^o+R_2^x &\leq& I(U_2,C_2;Y_1|X_1,Q) \label{eq:a4}\\
R_1^o+R_1^x &\leq& I(C_1;Z|U_1,U_2,C_2,Q) \label{eq:a5}\\
R_2^o+R_2^x &\leq& I(C_2;Z|U_1,U_2,C_1,Q) \label{eq:a6}\\
R_1^o+R_1^x + R_2^o+R_2^x &=& I(C_1,C_2;Z|U_1,U_2,Q) \label{eq:a7}\\
R_1^u+R_1^o &\leq& R_2^k \label{eq:a8}\\
R_2^u+R_2^o &\leq& R_1^k  \label{eq:a9}
\end{eqnarray}
\end{thm}
\begin{IEEEproof}
Please refer to Appendix~\ref{prf:FullDuplexDMC-R}.
\end{IEEEproof}

For $i=1,2$, $R_i^s$ denotes the rate of {\em physically} secure transmission for user $i$. i.e., the part of message $W_i$ that is secured using cooperative binning and channel prefixing only, $R_i^k$ denotes the rate of key transmission from user $i$ to the other user, $R_i^o$ denotes the rate of transmission of the open part of message $W_i$ that is secured using the secret key received from the other user in the previous block. The classical wiretap code~\cite{wyner} requires sacrificing part of the rate available for reliable communication, to exploit the secrecy advantage offered by the physical channel (in our case, the equivalent channel after inserting the channel prefix) in order to hide the message from the eavesdropper. The aforementioned part equals $R_i^o+R_i^x$ for user $i$. Note that the eavesdropper may be able to decode this part of message $W_i$, including the open part, but that will not violate the secrecy condition since this part is secured by the secret key received from the other user. The possibility of using a superposition code~\cite{nitbook} to transmit the physically secured message is allowed, where all nodes - including the eavesdropper - can identify the position of the cloud center, however, the part of the message conveyed through the cloud center is secured through the secret key received from the other user in the previous block, and in this case the rate of transmission for this part is given by $R_i^u$. The random variables $Q$ and $U$ denote the time sharing random variable and the cloud center of the superposition code, respectively.

Inequalities~\eqref{eq:a1}-~\eqref{eq:a4} follow from the reliable communication constraint,  and the conditions in~\eqref{eq:a5}-~\eqref{eq:a7} ensure that enough randomization is inserted through the wiretap code into the multiple access channel from the two legitimate nodes to the eavesdropper, such that the secrecy constraint is satisfied. Finally, the conditions in~\eqref{eq:a8}-~\eqref{eq:a9} follow from the fact that the entropy of the part of the message that is secured using the secret key received from the other user is bounded by the entropy of that key~\cite{Shannon:BSTJ:49}. Note that the role of key sharing evident from the above inequalities, is not to increase the sum rate, but to give complete freedom in distributing the the secrecy advantage offered by the two-way wiretap channel (after inserting the channel prefixes) between the two users.

\begin{rem}
The proposed coding scheme can be used to exchange
open messages (secured using the secret key) in addition to the physically secure ones between Alice and Bob, even through the cloud center of the superposition code.
More Specifically, the rate $R_i^u$ can be split into an open part $R_i^{uo}$ and a physically secured part $R_i^{us}$. Let $R_i^{\textrm{secret}}$ and
$R_i^{\textrm{open}}$ be the secret and open message rates
of transmitter $i=1,2$. Then, the proposed scheme readily achieves
the four-dimensional rate region given by the closure of
the union (over all input probability distributions) of the set of
rate tuples
$$
(R_1^{\textrm{secret}}=R_1^s+R_1^o+R_1^{us}, R_1^{\textrm{open}}=R_1^x+R_1^{uo},
R_2^{\textrm{secret}}=R_2^s+R_2^o+R_2^{us}, R_2^{\textrm{open}}=R_2^x+R_2^{uo}),
$$
with the non-negative rate tuples
$(R_1^{us},R_1^{uo},R_1^s,R_1^o,R_1^x,R_2^{us},R_2^{uo},R_2^s,R_2^o,R_2^x)$
satisfying \eqref{eq:a1}-\eqref{eq:a7} 
with $R_1^u = R_1^{us}+R_1^{uo}$, $R_2^u = R_2^{us}+R_2^{uo}$ and
$R_1^{us}+R_1^o \leq R_2^k$, 
$R_2^{us}+R_2^o \leq R_1^k$.
\end{rem}

One can immediately see that the region $\Rc$ does not lend itself to
simple computational approaches. Therefore, the rest of the section
will focus {\em primarily} on the following sub-region
$\Rc^F$.

\begin{thm}\label{thm:fullduplexdmc}
For the full-duplex DM-TWC-E,
\begin{equation*}
{\cal R}^F \triangleq \textrm{ closure of }
\left\{\underset{p\in{\cal P}^F}{\bigcup} {\cal R}^F(p)\right\}
\subseteq \Rc \subseteq \Cc^F,
\end{equation*}
where ${\cal P}^F$ denotes the set of all joint distributions
of the random variables $Q$, $C_1$, $C_2$, $X_1$, and $X_2$ satisfying
\begin{equation*}
P(q,c_1,c_2,x_1,x_2)=P(q)P(c_1|q)P(c_2|q)P(x_1|c_1)P(x_2|c_2)
\end{equation*}
and ${\cal R}^F(p)$ is the closure of all non-negative
rate tuples $(R_1,R_2)$ satisfying
\begin{eqnarray*}
R_1&\leq& I(C_1;Y_2|X_2,Q)\\
R_2&\leq& I(C_2;Y_1|X_1,Q)\\
R_1+R_2&\leq& I(C_1;Y_2|X_2,Q)+I(C_2;Y_1|X_1,Q)-I(C_1,C_2;Z|Q).
\end{eqnarray*}
\end{thm}
\begin{IEEEproof}
Please refer to Appendix~\ref{prf:fullduplexdmc}.
\end{IEEEproof}

Note that the above region, $\Rc^F$, is achievable without the need to use superposition coding, hence it is not clear to us whether the use of a superposition code is needed or not.
(Please refer to Remark~\ref{rem:RegionWithoutU1U2} in
Appendix~\ref{prf:fullduplexdmc}.)

\subsection{The Modulo-Two Channel}
To shed more light on the structural properties of our
achievable rate region, we now consider the special case of the full-duplex
modulo-$2$ two-way wiretap channel described by the
following set of input-output relations.
\begin{eqnarray*}
{\bf Y}_1&=&{\bf X}_1\oplus {\bf X}_2 \oplus {\bf N}_1\\
{\bf Y}_2&=&{\bf X}_1\oplus {\bf X}_2\oplus {\bf N}_2\\
{\bf Z}&=&{\bf X}_1\oplus {\bf X}_2 \oplus {\bf N}_e,
\end{eqnarray*}
where ${\bf N}_1$$=\{N_1(1),\ldots,N_1(n)\}$,
${\bf N}_2$$=\{N_2(1),\ldots,N_2(n)\}$, and
${\bf N}_e$$=\{N_e(1),\ldots,N_e(n)\}$ are the additive binary noise
vectors impairing Alice, Bob, and Eve, respectively.
The corresponding transition probabilities are given by:
$P({N_1}(t)=1)=\epsilon_1$, $P(N_2(t)=1)=\epsilon_2$, and
$P(N_e(t)=1)=\epsilon_e$ for $i=1,\ldots,n$.
The secrecy capacity region is denoted by $\Cc^{FM}$. In this special case, the transmitted codeword reduces to the
modulo-$2$ sum of a binning codeword and an independent
{\em prefix} noise component, i.e.,
\begin{eqnarray*}
{\bf X}_1&=&{\bf C}_1\oplus\bar{{\bf N}}_1 \\
{\bf X}_2&=&{\bf C}_2\oplus\bar{{\bf N}}_2,
\end{eqnarray*}
where ${\bf \bar{N}}_1$$=\{\bar{N}_1(1),\ldots,\bar{N}_1(n)\}$,
${\bf \bar{N}}_2$$=\{\bar{N}_2(1),\ldots,\bar{N}_2(n)\}$ are the
{\em prefix} noise vectors transmitted by Alice and Bob.
The components of these vectors are generated according to i.i.d.
distributions with the following marginals:
$P(\bar{N}_1(t)=1)=\bar{\epsilon}_1$,
$P(\bar{N}_2(t)=1)=\bar{\epsilon}_2$ for
$i=1,\ldots,n$. The binning codebooks, on the other hand,
are generated according to a uniform i.i.d. distribution.
We further define the following crossover probabilities
to describe the cascade of the prefix and original channels.
\begin{eqnarray*}
P(y_1\neq c_2|c_2)&=&\hat{\epsilon}_1 \triangleq
\epsilon_1(1-\bar{\epsilon}_2)+\bar{\epsilon}_2(1-\epsilon_1)\\
P(y_2\neq c_1|c_1)&=&\hat{\epsilon}_2 \triangleq
\epsilon_2(1-\bar{\epsilon}_1)+\bar{\epsilon}_1(1-\epsilon_2)\\
P(z\neq (c_1\oplus c_2)|c_1,c_2)&=&\hat{\epsilon}_e \triangleq
\epsilon_e(1-\bar{\epsilon}_{12})+\bar{\epsilon}_{12}(1-\epsilon_e)
\end{eqnarray*}
where, $\bar{\epsilon}_{12}=\bar{\epsilon}_2(1-\bar{\epsilon}_1)+\bar{\epsilon}_1(1-\bar{\epsilon}_2)$.
The need for the channel prefixes is evident in the case when the physical channel does not offer a secrecy advantage. For example, for the case when all channels are noiseless ($\epsilon_1=\epsilon_2=\epsilon_e=0$), no positive secrecy rates are achievable with only binning and key sharing. However, it is easy to see that the rates $(R_1,R_2)=(1,0)$ and $(0,1)$ are achievable with a choice of $(\bar{\epsilon}_1,\bar{\epsilon}_2)= (0,0.5)$ and $(0.5,0)$, respectively. Using the above notation, the achievable region
in Theorem~\ref{thm:fullduplexdmc} reduces to the region ${\cal R}^{FM}$ defined as follows.
\begin{cor}\label{cor:fullduplexmodulo}
For the full-duplex modulo-$2$ two-way wiretap channel
\begin{equation*}
{\cal R}^{FM} \triangleq \textrm{closure of the convex hull of}
\left\{\underset{p\in{\cal P}^{FM}}{\bigcup} {\cal R}^{FM}(p)\right\}
\subseteq \Cc^{FM},
\end{equation*}
where ${\cal P}^{FM}$ is defined as,
\begin{eqnarray}
{\cal P}^{FM} \triangleq &&\{(\bar{\epsilon}_1,\bar{\epsilon}_2): 0 \leq \bar{\epsilon}_1,\bar{\epsilon}_2 \leq 1 \notag\},
\end{eqnarray}
and ${\cal R}^{FM}(p)$ is the closure of all non-negative rate
tuples $(R_1,R_2)$ satisfying
\begin{eqnarray*}
R_1&\leq& 1-H(\hat{\epsilon}_2)\\
R_2&\leq& 1-H(\hat{\epsilon}_1)\\
R_1+R_2&\leq& 1+H(\hat{\epsilon}_e)
-H(\hat{\epsilon}_1)-H(\hat{\epsilon}_2).
\end{eqnarray*}
Moreover, our achievable region contains the two corner points of the {\bf secrecy capacity region}, namely
\begin{eqnarray}
\max\limits_{(R_1,0)\in\Cc} R_1 &=& 1-H(\epsilon_1), \textrm{ and}\notag\\
\max\limits_{(0,R_2)\in\Cc} R_2 &=& 1-H(\epsilon_2).\notag
\end{eqnarray}
\end{cor}
\begin{IEEEproof}
Please refer to Appendix~\ref{prf:fullduplexmodulo}.
\end{IEEEproof}

A few remarks are now in order.
\begin{enumerate}

\item The region in Corollary~\ref{cor:fullduplexmodulo} is strictly
larger than the ones reported in~\cite{TW-E,GGMAC-WT}, as demonstrated
by the numerical results of Fig.$~\ref{fig:modulo}$.
Here we compare our region with the one achieved by random binning
and key sharing only,
and channel prefixing only (\cite[Section 5]{TW-E}). The region
reported in~\cite[Theorem 2]{GGMAC-WT} can be achieved via binning without key sharing, hence,
is a {\bf strict} sub-region of Corollary~\ref{cor:fullduplexmodulo}.

\item The corner points of the
region in Corollary~\ref{cor:fullduplexmodulo} is
achieved by random binning and key sharing only if
$\epsilon_e > \max(\epsilon_1,\epsilon_2)$, and
achieved by only channel prefixing if
$\epsilon_e < \min(\epsilon_1,\epsilon_2)$.

\item The previous result identifies the separate role of channel
prefixing and binning. First, channel prefixing is used to
create an advantage of Alice and Bob over Eve via the {\bf joint}
optimization of $\bar{\epsilon}_1$ and $\bar{\epsilon}_2$.
Then, the binning codebooks are used to transform this
advantage into a secrecy gain for the two terminals.

\end{enumerate}

\subsection{The Gaussian Channel}\label{full_duplexGaussian}
In the full-duplex Gaussian setting, the channel is given by,
\begin{equation*}
{\bf Y}_1 = \sqrt{g_{11}} {\bf X}_1 +  {\bf X}_2+{\bf N}_1
\end{equation*}
\begin{equation*}
{\bf Y}_2 = {\bf X}_1+ \sqrt{g_{22}} {\bf X}_2+{\bf N}_2
\end{equation*}
\begin{equation*}
{\bf Z} = \sqrt{g_{e1}} {\bf X}_1+\sqrt{g_{e2}} {\bf X}_2+{\bf N}_e
\end{equation*}
where $g_{11}$, $g_{22}$, $g_{e1}$, and $g_{e2}$ are channel
coefficients, ${\bf N}_1$, ${\bf N}_2$, and ${\bf N}_e$ are i.i.d. noise
vectors with zero-mean unit-variance white Gaussian entries at
user $1$, user $2$, and Eve, respectively.
We assume the average power constraints given by
\begin{equation*}
\frac{1}{n}\sum_{t=1}^{n} (X_i(t))^2 \leq \rho_i, \textrm{ for } i=1,2.
\end{equation*}
The secrecy capacity of this channel is denoted by $\Cc^{FG}$.

We define $\gamma(x)\triangleq\frac{1}{2}\log(1+x)$ and
$h(X) = -\int f_X(x) \log f_X(x)$. 
The prefix to the channel from user $1$ to user $2$ is an additive white Gaussian noise channel with i.i.d. noise $\bar{N}_1 \sim {\cal N}(0,\rho_1^n)$, where the allocated power for user $1$ is distributed among the signal $C_1$ and the artificial noise $\bar{N}_1$. More specifically, $C_1 \sim {\cal N}(0,\rho_1^c)$, and $\rho_1^c+\rho_1^n = \rho_1 - \epsilon$, and the transmitted signal ${\bf X}_1 = {\bf C}_1 + {\bf \bar{N}}_1.$
By the weak law of large numbers, $\frac{1}{n}\sum_{t=1}^n (X_1(t))^2 \to
\rho_1-\epsilon$ as $n\to\infty$. ${\bf X}_2$ is constructed similarly to obtain the following.
\begin{cor}\label{cor:fullduplexgaussian}
For the full-duplex Gaussian two-way wiretap channel, the achievable rate region ${\cal R}^{FG}$ is given by,
\begin{center}
${\cal R}^{FG} \triangleq$ closure of the convex
hull of $\left\{\underset{p\in{\cal P}^{FG}}{\bigcup} {\cal R}^{FG}(p)\right\}
\subseteq \Cc^{FG}$,
\end{center}
where ${\cal P}^{FG}$ is defined as,
\begin{eqnarray*}
{\cal P}^{FG} \triangleq &&\{(\rho_1^c,\rho_1^n,\rho_2^c,\rho_2^n):
\rho_1^c+\rho_1^n \leq \rho_1, \notag
\rho_2^c+\rho_2^n \leq \rho_2 \},
\end{eqnarray*}
and ${\cal R}^{FG}(p)$ is the closure of all non-negative
rate tuples $(R_1,R_2)$ satisfying
\begin{equation*}
R_1 \leq \gamma\left(\frac{\rho_1^c }{1+\rho_1^n }\right)
\end{equation*}
\begin{equation*}
R_2 \leq \gamma\left(\frac{\rho_2^c }{1+\rho_2^n }\right)
\end{equation*}
\begin{eqnarray*}
R_1+R_2 &\leq &\gamma\left(\frac{\rho_1^c }{1+\rho_1^n }\right)
+\gamma\left(\frac{\rho_2^c }{1+\rho_2^n }\right)
-\gamma\left(\frac{\rho_1^c g_{e1}+ \rho_2^c g_{e2}}{1+\rho_1^n g_{e1}+ \rho_2^n g_{e2}}\right)
\end{eqnarray*}
\end{cor}
\begin{IEEEproof}
The proof follows by extending Theorem~\ref{thm:fullduplexdmc}
to continuous random variables, where we also set $|\Qc|=1$,
and use the convex hull operation. The tools needed to extend the probability of error and equivocation analysis are already available in the literature[e.g. see~\cite{GGMAC-WT} and~\cite{GGMAC-WT-Correction}]. 
\end{IEEEproof}

In Fig.~\ref{fig:gaussian}, we compare the region of Corollary~\ref{cor:fullduplexgaussian}
with the following special cases: 1) Both users implement cooperative binning and key sharing without channel prefixing
and 2) One of the users implements individual secrecy encoding \cite{wyner}, the
other helps only with channel prefixing. The same trends of the modulo-$2$ case are observed here
except for the fact that channel prefixing does not
achieve the two extreme points of ${\cal R}^{FG}$.
We note that the region
reported in~\cite[Theorem 2]{GGMAC-WT} can be achieved by
implementing binning without key sharing, and hence, is a sub-region of Corollary~\ref{cor:fullduplexgaussian}.
The scheme in \cite[Section V]{GGMAC-WT} is either binning only
at both users, or binning at one user and channel prefixing (jamming)
at the other user. The resulting regions are
subregions of Corollary~\ref{cor:fullduplexgaussian} (the first one is a subregion of the dashed region and the
second one is the dotted region in Fig.~$2$.). Next, we compare our results with that of~\cite{HeAndYener}.
Let,
\begin{equation*}
R_1^*\triangleq \max\limits_{\alpha\in[0,1]} \quad
\alpha \left[ \gamma(\rho_1)-
\left[ \gamma\left(\frac{g_{e1}\rho_1}{1+g_{e2}\rho_2}\right) -
\frac{1-\alpha}{\alpha} \left[
\gamma\left(\rho_2\right)
-
\gamma\left(\frac{g_{e2}\rho_2}{1+g_{e1}\rho_1}\right)
\right]^+
\right]^+
\right]^+
\end{equation*}
$R_2^*$ is obtained by reversing the indices above.
Then, the achievable rate region proposed in~\cite{HeAndYener}
is given by the convex hull of the following three points:
$$ [0,0],  [R_1^*, 0], \textrm{and }  [0, R_2^*].$$ We note that the region $\Rc^{FG}$ given in Corollary~\ref{cor:fullduplexgaussian} \emph{strictly} includes this one.
(The proof of the inclusion part is given in Appendix~\ref{prf:HeAndYener}.)
Fig.~\ref{fig:gaussian2} demonstrates the fact that the inclusion can be strict.
The same figure also includes the achievable region obtained by
\emph{backward key sharing only}.
In this scheme, users utilize only the one time pad scheme in a time
division manner where the node first receives a secret key
and then uses it to secure the message.
The corresponding region can be described as follows. Let
\begin{eqnarray*}
R_1^{\dagger} \triangleq \max_{\alpha\in[0,1]}\min\left\{
\alpha \gamma(\rho_1) ,
(1-\alpha)\left[
\gamma(\rho_2)-\gamma \left( \frac{g_{e2}\rho_2}{1+g_{e1}\rho_1} \right)
\right]^+
\right\}.
\end{eqnarray*}
$R_2^{\dagger}$ is obtained by reversing the indices above.
Then \emph{backward key sharing} achieves
the convex hull of the following three points:
$$ [0,0],  [R_1^{\dagger}, 0], \textrm{and }  [0, R_2^{\dagger}].$$
Note that, this is a subregion of $\Rc$ (given in Theorem~\ref{thm:FullDuplexDMC-R}),
in which $C_2$ is used to transmit secret key from user $2$ to user $1$,
and $U_1$ is utilized to transmit secret message in a one time pad
fashion. Comparing $R_1^{\dagger}$
and $R_1^*$ in Fig.~\ref{fig:gaussian2}, we can see that this scheme can achieve
higher rates than the ones reported in~\cite{HeAndYener}. We also remark that
this example is an evidence of the fact that the region in
Theorem~\ref{thm:FullDuplexDMC-R} strictly includes that of
Theorem~\ref{thm:fullduplexdmc}. (That is,
$\Rc^F \subsetneq \Rc$ as  $R_1^{\dagger}\notin \Rc^F$ but
$R_1^{\dagger}\in \Rc$ for the Gaussian channel.)
In summary, the region in Theorem~\ref{thm:FullDuplexDMC-R}
includes all the stated regions as special cases.

\section{Half-Duplex Channels}\label{half_duplex}
Our first step is to define
the following equivalent full-duplex model for the half-duplex channel.

\begin{defn}
For a given half-duplex channel governed by
$P(y_2,z|x_1)$, $P(y_1,z|x_2)$, $P(z|x_1,x_2)$, and $P(y_1)P(y_2)P(z)$ an
\emph{equivalent} full-duplex channel
$P^*(y_1,y_2,z|x_1,x_2)$ is defined as follows.

We allow the channel inputs to take the values
in $\Xc_i^*= \{ \Xc_i, ? \}$, where
$?$ represents the no transmission event.
Similarly the channel outputs take values in
$\Yc_i^*= \{ \Yc_i, ? \}$, where
$?$ represents the no reception event (due to
the half-duplex constraint).
Then, for the $t^{\textrm{th}}$ symbol time, the
full-duplex channel $P^*(y_1,y_2,z|x_1,x_2)$
is said to be in one of the following states:

$1$) $x_1(t)\in \Xc_1$, $x_2(t)=?$ :
User $1$ is transmitting, user $2$ is in no transmission state.

$2$) $x_1(t)=?$, $x_2(t)\in \Xc_2$ :
User $1$ is in no transmission state, user $2$ is transmitting.

$3$) $x_1(t)\in \Xc_1$, $x_2(t)\in \Xc_2$ :
Both users are transmitting.

$4$) $x_1(t)=?$, $x_2(t)=?$ :
Both users are in the no transmission state.

Accordingly, the channel $P^*(y_1,y_2,z|x_1,x_2)$ is given by
\begin{eqnarray*}\label{eq:pstar}
P^*(y_1,y_2,z|x_1,x_2) = \left\{
\begin{array}{c}
P(y_2,z|x_1,x_2=?) {\bf 1}_{\{y_1,?\}}, \textrm{ for state } 1 \\
P(y_1,z|x_1=?,x_2) {\bf 1}_{\{y_2,?\}}, \textrm{ for state } 2 \\
P(z|x_1,x_2) {\bf 1}_{\{y_1,?\}} {\bf 1}_{\{y_2,?\}}, \textrm{ for state } 3 \\
P(y_1,y_2,z|x_1,?,x_2=?) , \textrm{ for state } 4,
\end{array}
\right.
\end{eqnarray*}
where ${\bf 1}_{\{x,y\}}=1,$ if $x=y$ and ${\bf 1}_{\{x,y\}}=0,$ if $x \neq y$, and
$P(y_2,z|x_1,x_2=?)$, $P(y_1,z|x_1=?,x_2)$, $P(z|x_1,x_2)$, and
$P(y_1,y_2,z|x_1=?,x_2=?)$ are given by the half-duplex channel.
\end{defn}

Using this definition and our results for the full-duplex channel, we obtain the following result.

\begin{cor}[Deterministic Scheduling]\label{thm:HalfDuplexDeterministicScheduling}
The following region ${\cal R}^{H-D}$ is achievable for the half-duplex DM-TWC-E with deterministic scheduling.
\begin{equation*}
{\cal R}^{H-D} \triangleq \textrm{ the closure of }
\left\{\underset{P\in{\cal P}^{H},P_{s1}+P_{s2}=1}{\bigcup} {\cal R}^{H-D}(P)\right\},
\end{equation*}
where ${\cal P}^{H}$ denotes the set of all joint distributions
of the random variables $Q$, $C_1$, $C_2$, $X_1$, and $X_2$ satisfying
\begin{equation*}
P(q,c_1,c_2,x_1,x_2)=P(q)P(c_1|q)P(c_2|q)P(x_1|c_1)P(x_2|c_2),
\end{equation*}
${\cal R}^{H-D}(P)$ is the closure of all non-negative
rate tuples $(R_1,R_2)$ satisfying
\begin{eqnarray*}
R_1&\leq& P_{s1}I(C_1;Y_2|Q,\textrm{state 1})\\
R_2&\leq& P_{s2}I(C_2;Y_1|Q,\textrm{state 2})\\
R_1+R_2&\leq& P_{s1}[I(C_1;Y_2|Q,\textrm{state 1})-
I(C_1;Z|Q,\textrm{state 1})]^+\nonumber\\
&&+P_{s2}[I(C_2;Y_1|Q,\textrm{state 2})-I(C_2;Z|Q,\textrm{state 2})]^+,
\end{eqnarray*}
and the channel is given by $P^*(y_1,y_2,z|x_1,x_2)$ as defined in~\eqref{eq:pstar}.

\end{cor}

\begin{proof}
The proof follows by Theorem~\ref{thm:FullDuplexDMC-R} with the channel given
by $P^*(y_1,y_2,z|x_1,x_2)$. In each block
we randomly select a state $S=k$ with probability $P_{s,k}$, and replace $Q$ by $\{Q,S\}$, where
the random sequence $\sv$ represents the channel states
(and given to all nodes).
The achievable region can be represented
with the given description, where the inputs
are chosen such that we only utilize state $1$ and $2$
as the states $3$ and $4$ do not increase the
achievable rates.
\end{proof}

The previous region is achievable with a deterministic scheduling approach whereby the two users Alice and Bob {\em a-priori} agree on the schedule. Consequently, Eve is made aware of the schedule.  Now, in order to further confuse the eavesdropper, we propose a
\emph{novel randomized scheduling} scheme whereby, in each channel use, user $i$ will be in a transmission state
with probability $P_i$. Clearly, this approach will result in collisions, wasting some opportunities for using the channels. However, as established shortly, the gain resulting from confusing Eve about the source of each transmitted symbol will outweigh these inefficiencies in many relevant scenarios. To simplify our derivations, we assume that all the nodes
can identify perfectly state $4$ (no transmission
state). Furthermore, we also give Eve {\bf an additional advantage} by informing her of the symbol durations belonging to state $3$, and as a result we have the term $-P_1P_2 I(C_1,C_2;Z|Q,\text{state }3)$ in the sum rate constraint below. These assumptions are practical in the Gaussian
channel, where the users can use the received power levels to distinguish these states. The following result characterizes the corresponding achievable region.

\begin{cor}[Randomized Scheduling]\label{thm:HalfDuplexRandomScheduling}
The region ${\cal R}^{H}$ is achievable for the half-duplex DM-TWC-E with randomized scheduling.
\begin{equation*}
{\cal R}^{H} \triangleq \textrm{ closure of }
\left\{\underset{P\in{\cal P}^{H},0\leq P_1,P_2 \leq 1}
{\bigcup} {\cal R}^{H}(P)\right\},
\end{equation*}
where ${\cal P}^{H}$ denotes the set of all joint distributions
of the random variables $Q$, $C_1$, $C_2$, $X_1$, and $X_2$ satisfying
\begin{equation*}
P(q,c_1,c_2,x_1,x_2)=P(q)P(c_1|q)P(c_2|q)P(x_1|c_1)P(x_2|c_2),
\end{equation*}
${\cal R}^{H}(P)$ is the closure of all non-negative
rate tuples $(R_1,R_2)$ satisfying
\begin{eqnarray*}
R_1&\leq& P_1(1-P_2)I(C_1;Y_2|X_2,Q,\textrm{state 1})\\
R_2&\leq& (1-P_1)P_2I(C_2;Y_1|X_1,Q,\textrm{state 2})\\
R_1+R_2&\leq& P_1(1-P_2)I(C_1;Y_2|X_2,Q,\textrm{state 1})+(1-P_1)P_2I(C_2;Y_1|X_1,Q,\textrm{state 2})\nonumber\\
&& - P_1P_2I(C_1,C_2;Z|Q,\textrm{state 3})
- (P_1(1-P_2)+(1-P_1)P_2) I(C_1,C_2;Z|Q,\textrm{state 1 or 2}),
\end{eqnarray*}
and the channel is given by $P^*(y_1,y_2,z|x_1,x_2)$ as defined in~\eqref{eq:pstar}.
\end{cor}

\begin{proof}
Please refer to Appendix~\ref{prf:HalfDuplexRandomScheduling}.
\end{proof}

Similar to the full-duplex scenario, we now specialize our results to the modulo-$2$ case. We model this channel as
a \emph{ternary input} channel where the third input corresponds
to the no-transmission event. This way, the three nodes can identify
the symbol intervals when no one is transmitting. Therefore, those symbols will be identified and erased,
and the crossover probabilities corresponding to the other three
states are given by,
\begin{eqnarray*}
P(z\neq c_1|\textrm{only user $1$ is transmitting})&=&\epsilon_{e1} \triangleq
\epsilon_e(1-\bar{\epsilon}_1)+\bar{\epsilon}_1(1-\epsilon_e) \notag\\
P(z\neq c_2|\textrm{only user $2$ is transmitting})&=&\epsilon_{e2} \triangleq
\epsilon_e(1-\bar{\epsilon}_2)+\bar{\epsilon}_2(1-\epsilon_e)\notag\\
P(z\neq(c_1\oplus c_2)|\textrm{both users are transmitting})&=&\hat{\epsilon}_{e} \notag
\end{eqnarray*}
where $\hat{\epsilon}_e$ is given as in the previous section. Moreover, for some $\mu_1,\mu_2\in[0,1]$,
we define the followings,
\begin{eqnarray*}
P(y_1=1|\textrm{only user $2$ is transmitting})&=&\hat{\mu}_1
\triangleq \hat{\epsilon}_1(1-\mu_2)+\mu_2(1-\hat{\epsilon}_1)\\
P(y_2=1|\textrm{only user $1$ is transmitting})&=&\hat{\mu}_2
\triangleq \hat{\epsilon}_2(1-\mu_1)+\mu_1(1-\hat{\epsilon}_2)\\
P(z=1|\textrm{only user $1$ is transmitting})&=&\mu_{e1}
\triangleq \epsilon_{e1}(1-\mu_1)+\mu_1(1-\epsilon_{e1})\\
P(z=1|\textrm{only user $2$ is transmitting})&=&\mu_{e2}
\triangleq \epsilon_{e2}(1-\mu_2)+\mu_2(1-\epsilon_{e2})\\
P(z=1|\textrm{both users are transmitting})&=&\hat{\mu}_e
\triangleq \hat{\epsilon}_e(1-\mu_{12})+\mu_{12}(1-\hat{\epsilon}_e),
\end{eqnarray*}
where, $\hat{\epsilon}_1$ and $\hat{\epsilon}_2$ are given
as in the previous section, and $\mu_{12}=\mu_1(1-\mu_2)+\mu_2(1-\mu_1)$. Using these definitions, the following result is obtained.
\begin{prop}\label{prop:halfduplexmodulo}
The set of achievable rates for the half-duplex modulo-$2$
two-way wiretap channel ${\cal R}^{HM}$ is given by,
\begin{equation*}
{\cal R}^{HM} \triangleq \textrm{ closure of the convex hull of }
\left\{\underset{P\in{\cal P}^{HM}}{\bigcup} {\cal R}^{HM}(P)\right\},
\end{equation*}
where ${\cal P}^{HM}$ is defined as,
\begin{eqnarray*}
{\cal P}^{HM} \triangleq \{(\bar{\epsilon}_1,\bar{\epsilon}_2,\mu_1,\mu_2,P_1,P_2): 0 \leq \bar{\epsilon}_1,\bar{\epsilon}_2,\mu_1,\mu_2,P_1,P_2 \leq 1, \notag\},
\end{eqnarray*}
and ${\cal R}^{HM}(P)$ is the closure of all non-negative rate tuples $(R_1,R_2)$ satisfying
\begin{equation*}
R_1\leq P_1(1-P_2)(H(\hat{\mu}_2)-H(\hat{\epsilon}_2))
\end{equation*}
\begin{equation*}
R_2\leq P_2(1-P_1)(H(\hat{\mu}_1)-H(\hat{\epsilon}_1))
\end{equation*}
\begin{eqnarray*}
R_1+R_2 &\leq& P_1(1-P_2)(H(\hat{\mu}_2)-H(\hat{\epsilon}_2))+P_2(1-P_1)(H(\hat{\mu}_1)-H(\hat{\epsilon}_1))\notag\\
&&\:{-}\:P_1P_2(H(\hat{\mu}_e)-H(\hat{\epsilon}_e))\notag\\
&&\:{-}\:(P_1(1-P_2)+P_2(1-P_1))\notag\\
&&\:\:\:\bigg(H(\mu_{e1}d_1+\mu_{e2}d_2)-0.5 H(d_1 \epsilon_{e1} + d_2 \epsilon_{e2}) - 0.5 H(d_1 (1-\epsilon_{e1}) + d_2 \epsilon_{e2})\bigg)\notag,
\end{eqnarray*}
where
\begin{eqnarray*}
d_1&=&\frac{P_1(1-P_2)}{P_1(1-P_2)+P_2(1-P_1)},\textrm{ and}\\
d_2&=&1-d_1.
\end{eqnarray*}
\end{prop}

\begin{IEEEproof}
Please refer to Appendix~\ref{prf:halfduplexmodulo}.
\end{IEEEproof}

The advantage offered by {\em randomized scheduling} is best demonstrated in the following example. First, we observe that cooperative binning and channel prefixing scheme with {\em deterministic} scheduling fails to achieve a non-zero secrecy rate if Eve's channel is {\bf not} more noisy than the legitimate channels. Now, consider the noiseless case, i.e., $\epsilon_1=\epsilon_2=\epsilon_e=0$. By setting $\mu_1=\mu_2=P_1=P_2=0.5$, $\bar{\epsilon}_1=0$, and $\bar{\epsilon}_2=0.5$,
Proposition~\ref{prop:halfduplexmodulo} shows that the randomized
scheduling approach allows user $1$ to achieve a secure rate of $R_1=0.25-0.5(1-H(0.25))>0.$

The final step is to specialize the region to the Gaussian channel with half-duplex nodes. Eve is again assumed to {\em perfectly}
identify the no transmission and simultaneous transmission states.
We select codewords and jamming sequences as Gaussian (with
powers $\rho_i^c$ and $\rho_i^n$, respectively). In addition, to further
increase Eve's ambiguity, users jointly set $(\rho_i^c+\rho_i^n)g_{ei}$
to the same value $\rho_r$
(assuming the channel knowledge at both users). The following result is readily available.

\begin{prop}\label{prop:halfduplexgaussian}
The set of achievable rates for the half-duplex Gaussian two-way
wiretap channel ${\cal R}^{HG}$ is given by,
\begin{center}
${\cal R}^{HG} \triangleq$ closure of the convex hull of
$\left\{\underset{P\in{\cal P}^{HG}}{\bigcup} {\cal R}^{HG}(P)\right\}$
\end{center}
where ${\cal P}^{HG}$ is defined as,
\begin{eqnarray*}
{\cal P}^{HG} \triangleq &&\{(\rho_1^c,\rho_1^n,\rho_2^c,\rho_2^n,P_1,P_2):
0 \leq P_1,P_2 \leq 1, (\rho_1^c+\rho_1^n) g_{e1}=(\rho_2^c+\rho_2^n) g_{e2}=\rho_r, \notag\\
&&P_1(\rho_1^c+\rho_1^n) \leq \rho_1, \notag
P_2(\rho_2^c+\rho_2^n) \leq \rho_2 \},
\end{eqnarray*}
and ${\cal R}^{HG}(P)$ is the closure of all non-negative
rate tuples $(R_1,R_2)$ satisfying

\begin{equation*}
R_1 \leq P_1(1-P_2)\gamma\left(\frac{\rho_1^c}
{1+\rho_1^n }\right)
\end{equation*}
\begin{equation*}
R_2 \leq P_2(1-P_1)\gamma\left(\frac{\rho_2^c}{1+\rho_2^n}\right)
\end{equation*}
\begin{eqnarray*}
R_1+R_2\leq P_1(1-P_2)\gamma\left(\frac{\rho_1^c }{1+\rho_1^n }\right)
+P_2(1-P_1)\gamma\left(\frac{\rho_2^c }{1+\rho_2^n }\right)
+ h(Z|C_1,C_2)-h(Z),
\end{eqnarray*}
where
\begin{eqnarray*}
h(Z)-h(Z|C_1,C_2)&=&P_1P_2\gamma\left(\frac{\rho_1^c g_{e1}+\rho_2^c g_{e2}}{1+\rho_1^n g_{e1}+\rho_2^n g_{e2}}\right)
+(P_1(1-P_2)+P_2(1-P_1))\frac{1}{2}\log(2\pi e(1+\rho_r)) \notag\\
&&\:{-} (P_1(1-P_2)+P_2(1-P_1)) \int_{j=-\infty}^{\infty} \int_{i=-\infty}^{\infty} f_{C_1}(i) f_{C_2}(j) h(Z|i,j)d f_{C_1}d f_{C_2},\notag\\
\end{eqnarray*}
and
\begin{eqnarray*}
f_{Z|C_1,C_2}(z|i,j)&=&d_1 f(z;i,1+\rho_1^n g_{e1}) + d_2 f(z;j,1+\rho_2^n g_{e2}),\\
d_1&=&\frac{P_1(1-P_2)}{P_1(1-P_2)+P_2(1-P_1)},\\
d_2&=&1-d_1,
\end{eqnarray*}
and $f(x;\mu,\sigma^2)$ is the value at $x$ of the probability density function of a Gaussian random variable with mean $\mu$ and variance $\sigma^2$. 
\end{prop}
We remark that the ambiguity at Eve can be further increased by randomizing
the transmit power levels at the expense of more receiver complexity
(due to the non-coherent nature of the transmissions). We implemented
this randomization idea in the next section, where the complexity issue is resolved by using energy classifiers.

\section{Randomization for Secrecy: Practical Implementation}\label{experimental}
In this section, we study a more practical half-duplex Gaussian setting where the {\em constant} channel coefficients are determined by the distance-based path losses in a $2$-D geometric model. Our focus will be devoted to the symmetric case where the
two messages have the same rate. Without any loss of generality,
Alice and Bob are assumed to be located on the $x$-axis at opposite
ends of the origin and Eve is assumed to be located {\bf outside}
a circle centered around the origin of radius $r_E$ at an
angle $\theta$ of the $x$-axis (see Figure~\ref{fig:BANmodel}). This key assumption
faithfully models the spatial separation, between the
legitimate nodes and eavesdropper(s), which characterizes near field wireless networks like Body Area Networks (BAN) [see e.g.~\cite{BAN}]. The performance of the proposed secure randomized scheduling
communication scheme will be obtained as a function of $r_E$ and
the distance between Alice and Bob, i.e., $d_{AB}$. In the
discrete-time model, the signals received by the three nodes in the $t^{th}$
symbol interval are given by
\begin{eqnarray*}
Y_1(t)&=&{\bf 1}_{\{X_1(t),0\}}\nonumber
\left[G_A (d_{AA}^{-\alpha/2}X_1(t)e^{-jkd_{AA}}
+d_{AB}^{-\alpha/2}X_2(t)e^{-jkd_{AB}})+N_1(t)\right]\nonumber\\
Y_2(t)&=&{\bf 1}_{\{X_2(t),0\}}\nonumber
\left[G_B (d_{AB}^{-\alpha/2}X_1(t)e^{-jkd_{AB}}
+d_{BB}^{-\alpha/2}X_2(t)e^{-jkd_{BB}})+N_2(t)\right]\nonumber\\
Z(t)&=&G_E(d_{AE}^{-\alpha/2}X_1(t)e^{-jkd_{AE}}+
d_{BE}^{-\alpha/2}X_2(t)e^{-jkd_{BE}})\nonumber +N_e(t)\nonumber,
\end{eqnarray*}
where {\em k} is the wave number, $G_A$, $G_B$ and $G_E$ are
propagation constants which depend on the receive antenna gains,
and $\alpha$ is the path loss exponent which will be taken to $2$
as in the free space propagation scenario. (One can easily extend
our results for other scenarios with different path loss exponents.)
For further simplicity, we restrict ourselves to binary encoding implying that $X_1(t)\in\left\{-\sqrt{\rho(t)},0,\sqrt{\rho(t)}\right\}$, where
$\rho(t)$ is the instantaneous signal to noise ratio at unit distance
in the $t^{th}$ symbol interval if Alice decides to transmit.
$X_1(t)=0$ if Alice decides not to transmit. The same applies to $X_2(t)$.
$\rho(t)$ is selected randomly in the range $[\rho_{\textrm{min}},\rho_{\textrm{max}}]$,
by varying the transmit power,
according to a distribution that is known {\em a priori} to all nodes. The indicator function ${\bf 1}_{\{x,y\}}$ is defined as in Section~\ref{half_duplex}. In order to ensure the robustness of our results, we assume
that Eve employs a large enough receive antenna, i.e., $G_E>>1$,
such that her receiver has a high enough SNR and the additive noise effect in $Z$ can be ignored. We assume $G_A=G_B=1$, and a hard decision decoder at both the legitimate
receiver(s) and the eavesdropper. We consider a {\em memoryless}
classifier ${\cal C}$ used by Eve to
identify the origin of each received symbol, i.e., the decision
is based only on the power level of the observed symbol in the
current time interval.
Here, {\em $P_{m}$} and {\em $P_{f}$} represent the probability
of miss detection and false alarm, respectively. Furthermore,
we use $P_{e|m}$ to denote the probability of symbol error
given occurrence of the miss detection event.
Finally, we use the following notation:
$\phi(x) \triangleq \int\limits_{-\infty}^{x}\frac{1}
{\sqrt {2\pi}}e^{\frac{-t^{2}}{2}}dt$.

The deterministic scheduling paradigm is represented by a {\bf Time Division Multiplexing} scheme
whereby only a single message is transmitted in any given time
frame, and the legitimate receiver {\em jams} the channel with
random-content feedback symbols at random time intervals. More specifically, the
receiver will transmit a feedback symbol at any time interval
with probability $\beta$. This feedback
will result in erroneous outputs at the eavesdropper due to
its inability to identify the symbols corrupted by the random
feedback signal and erasures at the legitimate receiver due to the half-duplex constraint. As argued in~\cite{Lai2008}, this scheme is
capable of completely impairing Eve in modulo-additive channels.
In our {\em real-valued} channel, however, a simple energy
classifier based on the average received signal power~\cite{rssi-under}
can be used by Eve to differentiate between corrupted and
{\em non-jammed} symbols. To overcome this problem, we use
pre-determined distributions for the transmit power of both
the data symbols, {\bf $f_{1}$}, and feedback symbols
{\bf $f_{2}$}. This randomized power allocation strategy is
intended to increase the probability of {\em misclassification}
at Eve. The following result characterizes the achievable rate with this scheme.
\begin{thm}\label{thm:tdm}
Using the proposed TDM protocol with randomized feedback and power allocation,
the following secrecy rate is achievable at each user.
\begin{eqnarray*}
R_{s} = 0.5\max_{\beta,f_{1},f_{2}}
\left\{\min_{\theta,\cal{C}}\left\{[R_{M}-R_{E}]^{+}\right\}\right\},
\end{eqnarray*}
where
\begin{eqnarray*}
R_{M} = (1 - \beta)\left(1 - H\left(1 -
\phi\left(\sqrt \frac{\rho_{\textrm{min}}}{{d_{AB}}^{\alpha}}\right)\right)\right)
\end{eqnarray*}
\begin{eqnarray*}
R_{E} = \left(1 - \beta\left(1 - P_{m}\right) -
(1- \beta)P_{f}\right)\nonumber
\left(1 - H\left(\frac{\beta P_{m}P_{e|m}}
{1 - \beta\left(1 -P_{m}\right) - (1- \beta)P_{f}}\right)\right)
\end{eqnarray*}
\end{thm}
\begin{IEEEproof}
Please refer to Appendix~\ref{prf:tdm}.
\end{IEEEproof}

In the randomized scheduling approach, each node will transmit its message during randomly
selected time intervals, where a single node's transmitter is
active in any given time interval with probability $P_{t}$,
and the transmit power level is randomly selected according to
the distribution $f$. Consequently, there are four possible
states of both transmitters in any particular time interval $i$.
Due to our noiseless assumption, the eavesdropper's antenna
will easily identify {\em silence} intervals. Eve's challenge, however, is to differentiate between the
other three states. Let $A$ and $B$ represent the transmission event of Alice
and Bob, respectively. Similarly, $A^{c}$ and $B^{c}$ are the complementary events. Finally, we let
$E_1 \rightarrow E_2$ to denote the occurrence of
event $E_1$ and its classification by Eve as event $E_2$, and denote the probability of error given that
the event $(A,B)$ was mistaken for $(A,B^c)$ by the classifier
as $P_{e|(A,B)\rightarrow(A,B^{c})}$. The following is the achievable secrecy rate with the two-way randomization approach.
\begin{thm}\label{thm:interactive}
Using the \emph{two-way randomized scheduling and power
allocation protocol}, the following secrecy rate 
is achievable at each user.
\begin{equation*}
R_{s}=\max_{P_{t},f}(\min_{\theta,\cal C}([R_{M}-\max(R_{EA},R_{EB})]^{+})),
\end{equation*}
where
\begin{equation*}\label{eqn:rminteractive}
R_{M} = P_{t}\left(1 - P_{t}\right)\left(1 - H\left(1 - \phi\left(\sqrt \frac{\rho_{\textrm{min}}}{{d_{AB}}^{\alpha}}\right)\right)\right)
\end{equation*}
\begin{equation*}\label{eqn:rea}
R_{EA} = D_{A} \left(1 - H\left(\frac{P_{e}^{(EA)}}{D_{A}}\right)\right)
\end{equation*}
\begin{equation*}\label{eqn:reb}
R_{EB} = D_{B} \left(1 - H\left(\frac{P_{e}^{(EB)}}{D_{B}}\right)\right)
\end{equation*}
$\\D_{A} = P_{t}^{2}P_{(A,B)\rightarrow(A,B^{c})}+ P_{t}\left(1 - P_{t}\right)P_{(A^{c},B)\rightarrow(A,B^{c})}+ P_{t}\left(1 - P_{t}\right)\left(1 - P_{(A,B^{c})\rightarrow(A^{c},B)} - P_{(A,B^{c})\rightarrow(A,B)}\right)\\$
$\\D_{B} = P_{t}^{2}P_{(A,B)\rightarrow(A^{c},B)}+ P_{t}\left(1 - P_{t}\right)P_{(A,B^{c})\rightarrow(A^{c},B)}+ P_{t}\left(1 - P_{t}\right)\left(1 - P_{(A^{c},B)\rightarrow(A,B^{c})} - P_{(A^{c},B)\rightarrow(A,B)}\right)\\$
$\\P_{e}^{(EA)} = P_{t}^{2}P_{(A,B)\rightarrow(A,B^{c})}P_{e|(A,B)\rightarrow(A,B^{c})}+ 0.5P_{t}\left(1 - P_{t}\right)P_{(A^{c},B)\rightarrow(A,B^{c})}\\$
$\\P_{e}^{(EB)} = P_{t}^{2}P_{(A,B)\rightarrow(A^{c},B)}P_{e|(A,B)\rightarrow(A^{c},B)}+0.5P_{t}\left(1 - P_{t}\right)P_{(A,B^{c})\rightarrow(A^{c},B)}\\$
and $D_{A}$, $D_{B}$ represent the portion of symbols classified by Eve as being transmitted by Alice or Bob respectively.
\end{thm}
\begin{IEEEproof}
Please refer to Appendix~\ref{prf:interactive}.
\end{IEEEproof}
One can argue that the achievable secrecy rate increases as $r_E$ increases. The reason is that a large $r_E$ will impair Eve's ability to differentiate between the symbols transmitted by Bob and Alice. The following result characterizes the secrecy rate achievable in the asymptotic scenario when $r_{E} >> d_{AB}$.

\begin{cor}\label{cor:maxinteractive}
Let $R_{\textrm{max}}$ be the achievable
secrecy rate \emph{using the randomized scheduling and power allocation scheme} when $r_E\rightarrow\infty$.
Then,
\begin{equation}\label{eqn:upperbound}
R_{\textrm{max}} = \max_{P_{t}}([R_{M}-(1-(1-P_{t})^{2})(1-H(0.25))]^{+}),
\end{equation}
where
\begin{equation*}
R_{M} = P_{t}\left(1 - P_{t}\right)\left(1 - H\left(1 - \phi\left(\sqrt \frac{\rho_{\textrm{min}}}{{d_{AB}}^{\alpha}}\right)\right)\right)
\end{equation*}
\end{cor}
\begin{IEEEproof}
Please refer to Appendix~\ref{prf:maxinteractive}.
\end{IEEEproof}

\subsection{Numerical Results}
In our numerical examples, we assume a uniform power distribution
for both Alice and Bob, and a threshold-based energy classifier is
used by Eve. Because we assume that all channels are noiseless, Eve can successfully decode
the received symbols, corresponding to concurrent transmissions, as
the symbols with the higher received signal power.  Also, the received signal powers in all transmission scenarios are known a priori, where a transmission scenario is defined by the set of active transmitters and the selected power levels. Based on the received signal power, the transmission scenario is detected by Eve, and hence the set of active transmitters. In case two or more transmission scenarios result in the same received signal power, a random choice is made with equal probabilities given to all possible scenarios. To simplify the
calculations, we further assume that Alice and Bob use sufficient error
control coding to overcome the additive noise effect. More precisely,
Alice and Bob are assumed to use asymptotically optimal forward error
control coding and that their received SNR is above the minimal level
required to achieve arbitrarily vanishing probability of error.

Fig.$~\ref{fig:numerical}$ reports the achievable secrecy
rate $R_{s}$ of Theorems~\ref{thm:tdm} and~\ref{thm:interactive}
at different values for the distance ratio $\frac{d_{\textrm{min}}}
{d_{\textrm{max}}}$ ($d_{\textrm{min}}$ = $\min$($d_{AE},d_{BE})$,
$d_{\textrm{max}}$ = $\max$($d_{AE},d_{BE}$)).
A few remarks are now in order.
\begin{enumerate}
\item The two-way randomization scheme achieves
higher rates than the TDM scheme. The reason is the added ambiguity
at Eve resulting from the randomization in the scheduling algorithm.

\item The lower secrecy rates for smaller values of
$\frac{d_{\textrm{min}}}{d_{\textrm{max}}}$ is due to
Eve's enhanced ability to capture the symbols transmitted by the node closer to her.

\item The rates plotted in Fig.~$\ref{fig:numerical}$ were
found to be very close to those of a classifier that does not
erase any received symbols, i.e., transmission scenarios corresponding to concurrent transmissions are not considered.
\end{enumerate}

\subsection{Experimental Results}
We implemented our experiments on TinyOS \cite{tinyos-original} using TelosB motes
\cite{telos-datasheet}, which have a built-in CC2420 radio module
\cite{cc2420-datasheet}. The CC2420 module uses the IEEE
802.15.4 standards in the $2.4$~GHZ band ~\cite{802154-specifications}.
Our setup consists of four nodes, equivalent to
Alice, Bob, Eve, and a Gateway module. The Gateway acts as a link
between the sensor network and a PC running a java program.
Our experiment is divided into cycles. During each cycle,
the PC works as an orchestrator, {\em through the Gateway},
that determines, using a special message ({\em TRIGGER-MSG}), whether
Alice should send alone, Bob sends alone, or both send
concurrently. It also determines the power level used for transmission.
These decisions are based on the transmission probability $P_{t}$.
Upon receiving the broadcast TRIGGER-MSG, each trusted node
transmits a {\em DATA-MSG} while Eve will start to continuously read the value
in the Received Signal Strength Indicator (RSSI) register (the
RSSI value read by the CC2420 module is a moving average
of the last $8$ received symbols~\cite{cc2420-datasheet}.). Eve
then transfers the RSSI readings from the memory buffer to
the Gateway node which will forward them to the PC in an {\em RSSI-MSG}.
For each cycle, the java program stores the received RSSI readings for
further processing by the energy classifier (implemented in MATLAB).
When transmitting data messages (\emph{DATA-MSG}) from Alice or Bob,
each node constructs a random payload of 100 bytes using
the RandomMlcg component of TinyOS, which uses the Park-Miller
Minimum Standard Generator. Each symbol is {\it O-QPSK} modulated
\cite{802154-specifications} representing $4$ bits of the data.
We also had to remove the CSMA-CA mechanism from the CC2420
driver in order to allow both Alice and Bob to transmit concurrently.
Finally, it is worth noting that the orchestrator was used to overcome
the synchronization challenge in our experimental set-up. In practical
implementations, Bob (or Alice) could start jamming the
channel upon receiving the Start of Frame Delimiter (SFD).

In our implementation of the energy classifier, the discrete
nature of the transmit power levels is taken into consideration.
First, the eavesdropper was given the advantage of having the
classifier trained on a set of readings taken by running the
experiment in the same environment and at the same node locations
as those for which the classifier would be later used. In the
training phase, our classifier is given prior information on
the configuration, power levels selected for each node, and
the measured RSSI readings at each cycle. It then finds the
mean and variance of the measured RSSI values for each transmitted
power level for Alice and Bob when each of them sends alone in a cycle.
Any received symbol is classified as being transmitted by either of the
communicating nodes. This choice is based on our third observation
on the rates plotted in Fig.~$\ref{fig:numerical}$. When running
the classifier, a {\em maximum likelihood} rule is employed, where
the following expression is evaluated,
\begin{eqnarray*}\label{eqn:ml}
\frac{\max_{i} f_{A_{i}}(y)}{\max_{i} f_{B_{i}}(y)} \overset{A}{\underset{B}{\gtrless}} 1
\end{eqnarray*}
and the symbol is classified accordingly, where $f_{X_{i}}(y)$
is the value of the approximated Gaussian distribution of measured
RSSI values when source $X$ is the only transmitter with power
level $i$. In a practical implementation, the length of a cycle is
the duration of a single symbol, and hence, in our setup the classifier
bases its decision on a single RSSI reading. In evaluating the
classifier performance, we use the transmission scenario indicating the actual
status of the transmitters in each cycle and compare them with
the classification results to obtain the probability of each
possible misclassification event. We also assume that, in case of
concurrent transmission, Eve can correctly decode the symbol
received with the higher signal power, as suggested
in~\cite{Whitehouse2005}. This assumption is used to calculate the
values of $P_{e|(A,B)\rightarrow(A,B^{c})}$ and $P_{e|(A,B)\rightarrow(A^{c},B)}$.
We also use the same set of data to train and run a classifier for the
TDM protocol described above. Here, we
only consider cycles when Alice's transmitter is active, and consider
Bob's concurrent transmission as {\em jamming}.

Our experiments were conducted in a hallway environment, where only few
scatterers exist (only the wall structure). We train, run, and
evaluate our energy classifier, then use the resulting probabilities
in the rate expressions of Theorem~\ref{thm:tdm} and
Theorem~\ref{thm:interactive} to find the achievable
secrecy rates.
Figs.~\ref{fig:tdm} and~\ref{fig:interactive} report these results
in two representative configurations.
In the first, Alice and Bob are placed at the same location with
$d_{AE}=d_{BE}=20ft$, whereas $d_{AE}=1ft$ and $d_{BE}=20ft$
in the second. We note that the measured
difference of received signal power values from both transmitting
nodes was found to be $2$dB and $19$dB for Configurations $1$ and $2$,
respectively. This implies that the maximum rates in Fig.~\ref{fig:interactive}
and Fig.~\ref{fig:tdm} should be compared to the value of $R_s$ in
Fig.~\ref{fig:numerical} at $\frac{d_{\textrm{min}}}{d_{\textrm{max}}}=$
$0.79$ and $0.11$ respectively. We believe that this difference between
the theoretical and experimental results can be attributed to hardware
differences and the deviation of the actual channel from the simplistic
free space model used in our derivations. More specifically, we observe
that the maximum secrecy rates for the two-way randomized scheduling scheme
in our experimental results is slightly lower than those calculated numerically.
The reason is Eve's enhanced ability to distinguish between
the two sources of transmission due to the discrete nature of
the selected transmit power values. Nevertheless, the experimental
results establish the ability of our two-way randomized scheduling and power
allocation scheme to achieve perfect secrecy in practical near field
communication scenarios where the distance between Eve and
legitimate nodes will be larger than the inter-node distance,
{\bf even if Eve is equipped with a very large receive antenna}.

\section{Conclusion}\label{conclusion}
In this paper, we used the cooperative binning and channel prefixing approach to
obtain achievable secrecy rates for both the discrete memoryless
and Gaussian full-duplex two-way wiretap channels. In the proposed scheme, channel
prefixing is used to create \emph{an advantage} for the legitimate terminals
over the eavesdropper which is transformed by the binning codebooks into a
non-trivial secrecy rate region. A private key sharing and encryption was used to distribute the secure sum rate between the two users. We then introduced the idea of randomized
scheduling and established its fundamental role in the half-duplex two-way wiretap channel. Our
theoretical analysis revealed the ability of the proposed {\em
randomization} approach to achieve relatively high {\bf secure}
transmission rates under mild conditions on the eavesdropper
location. The ambiguity introduced at the eavesdropper by randomized scheduling was further validated by numerical results and extensive
experimental results using IEEE 802.15.4-enabled sensor boards in near field communication scenarios.

\section*{Acknowledgment}
The authors are thankful to C. Emre Koksal of The Ohio State University 
for insightful discussions.

\appendices

\section{Proof of Theorem ~\ref{thm:FullDuplexDMC-R}}
\label{prf:FullDuplexDMC-R}

First, we fix the probability density function $P(q)$, then generate a sequence ${\bf q}^{n'}$,
where the entries are i.i.d., and each entry is randomly
chosen according to $P(q)$. The sequence ${\bf q}^{n'}$ is
then given to all nodes before the communication takes place.

{\bf Codebook Generation:}

Consider user $i\in\{1,2\}$ that has a secret message
$w_i\in{\cal M}_i=\{1,2,...,M_i\}$, and a private
key $w_i^k\in{\cal M}_i^k=\{1,2,...,M_i^k\}$. For a given distribution
$P(u_i|q)$ and the sequence ${\bf q}$, generate $M_i^u$
i.i.d. sequences $\uv_i^{n'}(w_i^u)$, where $w_i^u\in[1,\cdots,M_i^u=2^{n'R_i^u}]$.
For each codeword
$\uv_i^{n'}(w_i^u)$, generate
$M_i^s M_i^k M_i^o M_i^x=2^{n'(R_i^s+R_i^k+R_i^o+R_i^x-\epsilon_0)}$ i.i.d.
sequences ${\bf c}_i^{n'}$, where $M_i=M_i^s
M_i^o M_i^u$, and $P({\bf c}_i^{n'}|{\bf u}_i^{n'})=\prod_{t=1}^{n'} P(c_i(t)|u_i(t))$.
Randomly distribute these into double indexed bins, where
each bin has $M_i^oM_i^x=2^{n'(R_i^o+R_i^x-\epsilon_0)}$ codewords,
and is indexed by the tuple $(w_i^s,w_i^k)$,
$w_i^s \in \{1,\cdots,M_i^s=2^{n'R_i^s}\}$, $w_i^o \in \{1,\cdots,M_i^o=2^{n'R_i^o}\}$, and
$w_i^x \in \{1,\cdots,M_i^x=2^{n'R_i^x}\}$. These codewords are represented by
$\cv_i^{n'}(w_i^u,w_i^s,w_i^k,w_i^o,w_i^x)$.

{\bf Encoding:} We use a block encoding scheme, where the full message is transmitted over $B$ blocks, each of length $n'$, and $n=n'B$. In the rest of the proof, we use bold face letters to represent vectors of block length $n'$. In each block, each user will transmit a private
key in addition to its message, and the other user will use this private key in the next block to secure
its message fully or in part. We omit the block indices for readability. In any given block, user $1$
will send the corresponding block messages of $w_1\in{\cal M}_1$ and the randomly selected $w_1^k\in{\cal M}_1^k$. The message index ($w_1$) is used to
select a tuple $(w_1^s,\tilde{w}_1^u,\tilde{w}_1^o)$, where $\tilde{w}_1^u$ and $\tilde{w}_1^o$ are encrypted into $w_1^u$ and $w_1^o$, respectively, using the private key
$\bar{w}_{2}^k=[\bar{w}_{2}^{k1},\bar{w}_{2}^{k2}]$ received from the other user in the previous block. In other words, let $\tilde{\bf{b}}_1^u$, $\tilde{\bf{b}}_1^o$, $\bf{b}_1^u$, $\bf{b}_1^o$, $\bar{\bf{b}}_2^{k1}$, and
$\bar{\bf{b}}_2^{k2}$ be the binary representations of $\tilde{w}_1^u$, $\tilde{w}_1^o$, $w_1^u$, $w_1^o$, $\bar{w}_2^{k1}$, and $\bar{w}_2^{k2}$ respectively.
Then, $\bf{b}_1^u = \tilde{\bf{b}}_1^u \oplus \bar{\bf{b}}_2^{k1}$,
and $\bf{b}_1^o = \tilde{\bf{b}}_1^o \oplus \bar{\bf{b}}_2^{k2}$. Here,
$w_1^u$ is used to select the cloud center of the super position coding (see, e.g.,~\cite{CT}),
$(w_1^s,w_1^k)$ is used to select the bin index, and the
codeword index within the bin is given by $(w_1^o,w_1^x)$, where $w_1^x$ is randomly selected according to a uniform distribution. (Note that, due to one time pad,
$w_1^o$ is also uniformly distributed.)
Thus the corresponding codeword ${\bf c}_1^{n'}(w_1^u,w_1^s,w_1^k,w_1^o,w_1^x)$ is selected. Then, the channel
input, ${\bf x}_1^{n'}$, is generated using the distribution $P(x_1|c_1)$. A similar encoding scheme is
employed at user $2$.
As the messages transmitted in different blocks are independent, satisfying the reliability and security constraints for each block guarantees their application for all messages transmitted in an arbitrarily large number of blocks.

{\bf Decoding:}

Consider a message ${\bf y}_1^{n'}$ received at the receiver of
user $1$. Let $A_{1,\epsilon}^{n'}$ be the set of {\em weakly} typical
$({\bf q}^{n'},\uv_2^{n'}(w_2^u),$
${\bf c}_2^{n'}(w_2^u,w_2^s,w_2^k,w_2^o,w_2^x),{\bf y}_1^{n'})$ sequences.
As ${n'}\rightarrow \infty$, the decoder will select $(w_2^u,w_2^s,w_2^k,w_2^o,w_2^x)$
such that,
\begin{equation*}
({\bf q}^{n'},\uv_2^{n'}(w_2^u),{\bf c}_2^{n'}(w_2^u,w_2^s,w_2^k,w_2^o,w_2^x),
{\bf y}_1^{n'},{\bf x}_1^{n'})\in A_{1,\epsilon}^{n'}
\end{equation*}
if such a tuple exists and is unique. Otherwise, the decoder
declares an error. Note that the decoder's estimate $\hat{w}_2$ is determined by $(w_2^s,w_2^u,w_2^o,\bar{w}_1^k)$, where $\bar{w}_1^k$ is the private key sent by user $1$ in the previous block. Decoding at receiver $2$ is symmetric and
can be described by reversing the indices $1$ and $2$ above.

{\bf Probability of Error Analysis:}

It follows by the proof of the capacity of the point to point DMC~\cite{Shannon:BSTJ:48} that for any given $\epsilon>0$, receiver $1$ can decode the corresponding messages
with $P_{e,2}<\epsilon$ for sufficiently large ${n'}$, if
\begin{eqnarray}
R_2^s+R_2^k+R_2^o+R_2^x &\leq& I(C_2;Y_1|X_1,U_2,Q) \label{eq:cond1a} \\
R_2^u+R_2^s+R_2^k+R_2^o+R_2^x &\leq& I(U_2,C_2;Y_1|X_1,Q) \label{eq:cond1b}
\end{eqnarray}
By symmetry, a similar condition applies to
receiver $2$ to have $P_{e,1}<\epsilon$, i.e.,
\begin{eqnarray}
R_1^s+R_1^k+R_1^o+R_1^x &\leq& I(C_1;Y_2|X_2,U_1,Q) \label{eq:cond2a} \\
R_1^u+R_1^s+R_1^k+R_1^o+R_1^x &\leq& I(U_1,C_1;Y_2|X_2,Q) \label{eq:cond2b}
\end{eqnarray}

{\bf Equivocation Computation:}
Consider the following argument.
\begin{eqnarray}\label{eq:equivocation}
H(W_1^k,W_1^s,W_2^k,W_2^s|{\bf Z})&\overset{(a)}{\geq}& H(W_1^k,W_1^s,W_2^k,W_2^s|{\bf Z},\Um_1, \Um_2, {\bf Q})\notag
\\&=&H(W_1^k,W_1^s,W_2^k,W_2^s,{\bf Z}|\Um_1, \Um_2, {\bf Q})-H({\bf Z}|\Um_1, \Um_2, {\bf Q})\notag
\\&=&H(W_1^k,W_1^s,W_2^k,W_2^s,{\bf C}_1,{\bf C}_2,{\bf Z}|\Um_1, \Um_2, {\bf Q})
-H({\bf Z}|\Um_1, \Um_2, {\bf Q})\notag
\\&&{-}\:H({\bf C}_1,{\bf C}_2|W_1^k,W_1^s,W_2^k,W_2^s,{\bf Z},\Um_1, \Um_2, {\bf Q})\notag\\
&=&H({\bf Z}|{\bf C}_1,{\bf C}_2,W_1^k,W_1^s,W_2^k,W_2^s,\Um_1, \Um_2, {\bf Q})\notag\\
&&{+}\:H(W_1^k,W_1^s,W_2^k,W_2^s,{\bf C}_1,{\bf C}_2|\Um_1, \Um_2, {\bf Q})\notag\\
&&{-}\:H({\bf Z}|\Um_1, \Um_2, {\bf Q})-H({\bf C}_1,{\bf C}_2|W_1^k,W_1^s,W_2^k,W_2^s,{\bf Z},\Um_1, \Um_2, {\bf Q})
\notag
\\&\overset{(b)}{=}&[H({\bf Z}|{\bf C}_1,{\bf C}_2,\Um_1, \Um_2, {\bf Q})
-H({\bf Z}|\Um_1, \Um_2, {\bf Q})]+H({\bf C}_1,{\bf C}_2|\Um_1, \Um_2, {\bf Q})\notag\\
&&{-}\:H({\bf C}_1,{\bf C}_2|W_1^k,W_1^s,W_2^k,W_2^s,{\bf Z},\Um_1, \Um_2, {\bf Q})\notag
\\&\overset{(c)}{\geq}&-{n'}I(C_1,C_2;Z|U_1,U_2,Q)-{n'}\epsilon_1
+H({\bf C}_1,{\bf C}_2|\Um_1, \Um_2, {\bf Q})\notag\\
&&{-}\:H({\bf C}_1,{\bf C}_2|W_1^k,W_1^s,W_2^k,W_2^s,{\bf Z},\Um_1, \Um_2, {\bf Q}),
\end{eqnarray}
where (a) follows from the fact that conditioning does not
increase the entropy, (b) follows from the fact that, given
$\Um_1, \Um_2, {\bf Q}$, $(W_1^k,W_1^s,W_2^k,W_2^s)\rightarrow
({\bf C}_1,{\bf C}_2)\rightarrow ({\bf Z})$ is a Markov Chain,
and (c) follows from
$I({\bf C}_1,{\bf C}_2;{\bf Z}|\Um_1, \Um_2, {\bf Q})$ $\leq {n'} I(C_1,C_2;Z|U_1,U_2,Q)+{n'}\epsilon_1$
with $\epsilon_1\to 0$ as ${n'}\to \infty$ for a discrete memoryless
channel (see, e.g.,~\cite[Lemma 8]{wyner}).

Here,
\begin{equation}\label{eq:prf1}
H({\bf C}_1,{\bf C}_2|\Um_1, \Um_2, {\bf Q})={n'}(R_1^k+R_1^s+R_1^o+R_1^x+R_2^k+R_2^s+R_2^o+R_2^x-2\epsilon_0),
\end{equation}
as, given $(\Um_1, \Um_2, {\bf Q})=(\uv_1, \uv_2, {\bf q})$, the tuple $({\bf C}_1,{\bf C}_2)$ has
$2^{{n'}(R_1^k+R_1^s+R_1^o+R_1^x+R_2^k+R_2^s+R_2^o+R_2^x-2\epsilon_0)}$ possible values each
with equal probability, and,
\begin{equation*}
H({\bf C}_1,{\bf C}_2|W_1^k=w_1^k,W_1^s=w_1^s,W_2^k=w_2^k,W_2^s=w_2^s,{\bf Z},
\Um_1=\uv_1, \Um_2=\uv_2,{\bf Q}={\bf q})\leq {n'}\epsilon_2
\end{equation*}
for $\epsilon_2\rightarrow0$ as ${n'}\rightarrow\infty$. This follows from
the Fano's inequality, as the eavesdropper can decode the randomization indices $(w_1^o,w_1^x,w_2^o,w_2^x)$  given $(w_1^k,w_1^s,w_2^k,w_2^s)$ if the following conditions are satisfied.
\begin{equation}\label{eq:fcondronex}
R_1^o+R_1^x \leq I(C_1;Z|C_2,U_1,U_2,Q)
\end{equation}
\begin{equation}\label{eq:fcondrtwox}
R_2^o+R_2^x \leq I(C_2;Z|C_1,U_1,U_2,Q)
\end{equation}
\begin{equation}\label{eq:fcondrsumx}
R_1^o+R_1^x+R_2^o+R_2^x\leq I(C_1,C_2;Z|U_1,U_2,Q)
\end{equation}
By averaging over $W_1^k$, $W_1^s$, $W_2^k$, $W_2^s$,
$\Um_1$, $\Um_2$, and ${\bf Q}$, we obtain
\begin{equation}\label{eq:prf2}
H({\bf C}_1,{\bf C}_2|W_1^k,W_1^s,W_2^k,W_2^s,{\bf Z},\Um_1,\Um_2,{\bf Q})\leq {n'}\epsilon_2,
\end{equation}
Now, once we set,
\begin{equation}\label{eq:prf3}
R_1^o+R_1^x+R_2^o+R_2^x=I(C_1,C_2;Z|U_1,U_2,Q),
\end{equation}
and combine ~\eqref{eq:equivocation},~\eqref{eq:prf1},~\eqref{eq:prf2},
and~\eqref{eq:prf3}, we obtain
\begin{equation*}\label{eq:sec1}
\frac{1}{n'}H(W_1^k,W_1^s,W_2^k,W_2^s|{\bf Z})\geq R_1^k+R_1^s+R_2^k+R_2^s-(\epsilon_1+\epsilon_2+2\epsilon_0)
\end{equation*}
and $(\epsilon_1+\epsilon_2+2\epsilon_0)\rightarrow0$
as ${n'}\rightarrow\infty$.

Since $\bar{w}_2^k$ ($\bar{w}_1^k$) is used as a private key to secure the
part of the message carried in $w_1^u, w_1^o$ ($w_2^u, w_2^o$, respectively)
with the one-time-padded scheme, the secrecy constraint
\begin{equation*}
\frac{1}{n'}H(W_1,W_2|{\bf Z})\geq R_1+R_2-\epsilon
\end{equation*}
is satisfied (see~\cite{Shannon:BSTJ:49}) if
\begin{eqnarray}
R_1^u+R_1^o &\leq& R_2^k\label{eq:cond5a}\\
R_2^u+R_2^o &\leq& R_1^k\label{eq:cond5b}
\end{eqnarray}
where we set $R_1=R_1^u+R_1^o+R_1^s$ and $R_2=R_1^u+R_2^o+R_2^s$.

Finally, we note that $R_1^u=R_2^u=R_1^o=R_2^o=0$ for the first block.
However, the impact of this condition on the achievable rate
diminishes as the number of blocks $B \rightarrow \infty$.
The region achieved by the proposed scheme is given by
\eqref{eq:cond1a}, \eqref{eq:cond1b}, \eqref{eq:cond2a}, \eqref{eq:cond2b},
\eqref{eq:fcondronex}, \eqref{eq:fcondrtwox}, \eqref{eq:fcondrsumx},
\eqref{eq:cond5a}, and \eqref{eq:cond5b}.


\section{Proof of Theorem ~\ref{thm:fullduplexdmc}}
\label{prf:fullduplexdmc}

For a given distribution $p\in\Pc^F$, let
$$I_6\triangleq I(C_1;Y_2|X_2,Q)-I(C_1;Z|Q),$$
$$I_7\triangleq I(C_2;Y_1|X_1,Q)-I(C_2;Z|Q),$$
and
$$I_8\triangleq I(C_1;Y_2|X_2,Q)+I(C_2;Y_1|X_1,Q)-I(C_1,C_2;Z|Q).$$

If $I_8<0$, we set $R_1=R_2=0$. Hence, we only focus on
cases for which $I_8\geq0$. This implies that $I_6\geq 0$
and/or $I_7\geq 0$. (As $I_6<0$ and $I_7<0$ implies that
$I_8<0$.) We detail the proof for the following cases.

\textbf{Case 1:} $I_6\geq 0$ and $I_7\geq 0$ for the given $p\in\Pc^F$.

We set $U_1,U_2$ as deterministic and $R_1^{u}=R_2^{u}=0$
in Theorem~\ref{thm:FullDuplexDMC-R}, and obtain that
\begin{eqnarray}
R_1^s+R_1^k+R_1^o+R_1^x &\leq& I(C_1;Y_2|X_2,Q)\triangleq I_1 \label{eq:Appendix1Eq11}\\
R_2^s+R_2^k+R_2^o+R_2^x &\leq& I(C_2;Y_1|X_1,Q)\triangleq I_2 \label{eq:Appendix1Eq12}\\
R_1^o+R_1^x &\leq& I(C_1;Z|C_2,Q)\triangleq I_3 \label{eq:Appendix1Eq13}\\
R_2^o+R_2^x &\leq& I(C_2;Z|C_1,Q)\triangleq I_4 \label{eq:Appendix1Eq14}\\
R_1^o+R_1^x + R_2^o+R_2^x &=& I(C_1,C_2;Z|Q)\triangleq I_5 \label{eq:Appendix1Eq15}\\
R_1^o &\leq& R_2^k \label{eq:Appendix1Eq16}\\
R_2^o &\leq& R_1^k  \label{eq:Appendix1Eq17}
\end{eqnarray}

As $I_6\geq 0$, $I_7\geq 0$, and $I_8\geq 0$, we can
choose the rates as follows:
\begin{itemize}
  \item If $I(C_2;Y_1|X_1,Q)\geq I(C_2;Z|C_1,Q)$, then we choose

  $R_1^k=0$,
  $R_1^o=R_2^k$,
  $R_1^x=\left[I(C_1,C_2;Z|Q)-I(C_2;Y_1|X_1,Q)\right]^+$,

  $R_1^s=I(C_1;Y_2|X_2,Q)-R_2^k-\left[I(C_1,C_2;Z|Q)-I(C_2;Y_1|X_1,Q)\right]^+$,

  $R_2^k=I(C_1;Z|Q)-\left[I(C_1,C_2;Z|Q)-I(C_2;Y_1|X_1,Q)\right]^+$,
  $R_2^o=0$,
  $R_2^x=I(C_1,C_2;Z|Q)-R_2^k-R_1^x$,

  $R_2^s=\left[I(C_2;Y_1|X_1,Q)-I(C_1,C_2;Z|Q)\right]^+$.
  \item If $I(C_2;Y_1|X_1,Q)<I(C_2;Z|C_1,Q)$, then we choose

  $R_1^s=I(C_1;Y_2|X_2,Q)-I(C_1,C_2;Z|Q)+I(C_2;Y_1|X_1,Q)$,
  $R_1^x=I(C_1,C_2;Z|Q)-R_2^x$,
  $R_2^x=I(C_2;Y_1|X_1,Q)$, and the remaining rates equal to zero.
\end{itemize}

These choice of \emph{non-negative} rates satisfy conditions in
\eqref{eq:Appendix1Eq11}-\eqref{eq:Appendix1Eq17}, and hence
we can achieve the rate pair
$$(R_1=I_1-[I_5-I_2]^+,R_2=[I_2-I_5]^+).$$ Similarly,
by reversing the indices above, the rate pair
$$(R_1=[I_1-I_5]^+,R_2=I_2-[I_5-I_1]^+)$$
is achievable. Now, combining these two achievable points
we obtain the following achievable region:
The set of non-negative ($R_1$, $R_2$) pairs satisfying
\begin{eqnarray*}
R_1 &\leq& I_1\label{eq:finalresult1}\\
R_2 &\leq& I_2\label{eq:finalresult2}\\
R_1+R_2 &\leq& I_1+I_2-I_5\label{eq:finalresult3}
\end{eqnarray*}
are achievable.

\textbf{Case 2:} $I_6\geq 0$ and $I_7< 0$ for the given $p\in\Pc^F$.

We set $U_1$ and $C_2$ as deterministic and choose
the following rates in Theorem~\ref{thm:FullDuplexDMC-R}
(other rates are chosen to be $0$).
\begin{eqnarray*}
R_1^{k}&=&I(C_1;Y_2|X_2,Q)-I(C_1;Z|U_2,Q)-R_1^{s}\\
R_1^{s}&\leq& I(C_1;Y_2|X_2,Q)-I(C_1;Z|U_2,Q)\\
R_1^{x}&=&I(C_1;Z|U_2,Q)\\
R_2^{u}&=&\min\{I(U_2;Y_1|X_1,Q),R_1^{k}\}
\end{eqnarray*}
For the given $p\in\Pc^F$
with $I_6\geq 0$ and $I_7< 0$, the following region is achievable.
\begin{eqnarray*}
R_1&\leq&I(C_1;Y_2|X_2,Q)\\
R_2&\leq&I(U_2;Y_1|X_1,Q)\\
R_1+R_2&\leq&I(C_1;Y_2|X_2,Q)-I(C_1;Z|U_2,Q)
\end{eqnarray*}
Note that the above region is the same as the one in the theorem statement, with the random variable $U_2$ taking the role of $C_2$.
\textbf{Case 3:} $I_6< 0$ and $I_7\geq 0$ for the given $p\in\Pc^F$.

Reversing the indices everywhere in case 2 above, we obtain
the following achievable region
\begin{eqnarray*}
R_1&\leq&I(C_1;Y_2|X_2,Q)\\
R_2&\leq&I(C_2;Y_1|X_1,Q)\\
R_1+R_2&\leq&I(C_2;Y_1|X_1,Q)-I(C_2;Z|C_1,Q)
\end{eqnarray*}

Combining the above cases completes the proof.

\begin{rem}\label{rem:RegionWithoutU1U2}
The above scheme either uses the one time padded private key as one of the two selectors for the randomization index (Case 1), or does not employ the random binning coding scheme and only uses the private key at one of the user (User 2 in Case 2, and User 1 in Case 3). Hence, no superposition coding is present. We should also note that the achievable
rates proved above in Cases 2 and 3, can be higher than that of the statement. However, as already mentioned, we only use this Theorem as a simple special case of Theorem~\ref{thm:FullDuplexDMC-R}.
\end{rem}


\section{Proof of Corollary~\ref{cor:fullduplexmodulo}}
\label{prf:fullduplexmodulo}

We set $|{\cal Q}|=1$ in Theorem~\ref{thm:fullduplexdmc}
and take the convex hull of the achievable rates.
We compute the following terms.
\begin{eqnarray}\label{eq:fmcondone}
I(C_1;Y_2|X_2,Q)&=&H(Y_2|X_2)-H(Y_2|C_1,X_2)\notag
\\&\leq&1-H(\hat{\epsilon}_2)
\end{eqnarray}
\begin{eqnarray}\label{eq:fmcondtwo}
I(C_2;Y_1|X_1,Q)&=&H(Y_1|X_1)-H(Y_1|C_2,X_1)\notag
\\&\leq&1-H(\hat{\epsilon}_1)
\end{eqnarray}
\begin{eqnarray*}\label{eq:fmeqone}
I(C_1;Y_2|X_2,Q)+I(C_2;Y_1|X_1,Q)-I(C_1,C_2;Z|Q)&=&
\left(H(Y_1|X_1)+H(Y_2|X_2)-H(Z)\right)\notag\\&&\:{+}\left(H(Z|C_1,C_2)-H(Y_1|C_2,X_1)-H(Y_2|C_1,X_2)\right)
\notag
\end{eqnarray*}
By noting that,
\begin{eqnarray*}\label{eq:fmeqtwo}
H(Y_1|X_1)+H(Y_2|X_2)-H(Z) &=& \left(H(X_2\oplus N_1)+H(X_1\oplus N_2)-H(X_1 \oplus X_2 \oplus N_e)\right)\notag
\\&\overset{(a)}{=}& H(X_2\oplus N_1)+H(X_1\oplus N_2)-H(X_2\oplus N_1 \oplus X_1\oplus N_2 \oplus \hat{N}_e)\notag
\\&\overset{(b)}{\leq}&H(X_2\oplus N_1)+H(X_1\oplus N_2)-H(X_2\oplus N_1 \oplus X_1\oplus N_2)\notag
\\&=&H(X_2\oplus N_1)+H((X_2\oplus N_1 \oplus X_1\oplus N_2)|(X_2\oplus N_1))
\notag\\&&\:{-}H(X_2\oplus N_1 \oplus X_1\oplus N_2)\notag
\\&=& H((X_2\oplus N_1), (X_2\oplus N_1 \oplus X_1\oplus N_2)) - H(X_2\oplus N_1 \oplus X_1\oplus N_2)\notag
\\&=& H((X_2\oplus N_1)|(X_2\oplus N_1 \oplus X_1\oplus N_2))\notag
\\&\leq& 1
\end{eqnarray*}
where (a) follows by setting $\hat{N}_e=N_1 \oplus N_2 \oplus N_e$, (b) follows from the fact that conditioning does not increase entropy, we conclude that,
\begin{equation}\label{eq:fmcondthree}
I(C_1;Y_2|X_2,Q)+I(C_2;Y_1|X_1,Q)-I(C_1,C_2;Z|Q)\leq 1+H(\hat{\epsilon}_e)
-H(\hat{\epsilon}_1)-H(\hat{\epsilon}_2),
\end{equation}

The proof is complete by combining the terms in
~\eqref{eq:fmcondone}, ~\eqref{eq:fmcondtwo}, and ~\eqref{eq:fmcondthree}
with Theorem~\ref{thm:fullduplexdmc}. We note that equality applies in the three 
mentioned terms when the variables $C_1,C_2$ are drawn from the
uniform distribution over $\{0,1\}$.


\section{The region $\Rc^{FG}$ includes that of~\cite{HeAndYener}}
\label{prf:HeAndYener}

We utilize the time sharing parameter as follows.
Let $\Qc=\{1,2\}$, where $q=1$ with prob. $(1-\alpha)$ and
$q=2$ with prob. $\alpha$. The remaining distributions are as follows.
\begin{itemize}
    \item For $q=1$, we set $C_1$ as deterministic and
    $X_1=\bar{N}_1$ for channel prefixing.
    $C_2$ and $\bar{N}_1$ are generated with full powers $P_2$
    and $P_1$, respectively.
    \item For $q=2$, we set $C_2$ as deterministic and
    $X_2=\bar{N}_2$ for channel prefixing.
    $C_1$ and $\bar{N}_2$ are generated with full powers $P_1$
    and $P_2$, respectively.
\end{itemize}

With this choice the region in Theorem~\ref{cor:fullduplexgaussian}
reduces to the following:
\begin{eqnarray*}
R_1&\leq& I(C_1;Y_2|X_2,Q)=\alpha \gamma(P_1)\\
R_2&\leq& I(C_2;Y_1|X_1,Q)=(1-\alpha) \gamma(P_2)\\
R_1+R_2&\leq& I(C_1;Y_2|X_2,Q)+I(C_2;Y_1|X_1,Q)
-I(C_1,C_2;Z|Q) \\
&=& \alpha \gamma(P_1)+(1-\alpha) \gamma(P_2)
- \alpha \gamma\left(\frac{g_{e1}P_1}{1+g_{e2}P_2}\right)
- (1-\alpha) \gamma\left(\frac{g_{e2}P_2}{1+g_{e1}P_1}\right)
\end{eqnarray*}

Let
\begin{equation*}
R_K \triangleq \gamma(P_2)-
\gamma\left(\frac{g_{e2}P_2}{1+g_{e1}P_1}\right),
\end{equation*}
and
\begin{eqnarray*}
R_1(\alpha) \triangleq \left[ \alpha \gamma(P_1)-
\left[ \alpha \gamma\left(\frac{g_{e1}P_1}{1+g_{e2}P_2}\right) -
(1-\alpha) R_K
\right]^+
\right]^+.
\end{eqnarray*}
If $R_K\leq 0$, then $R_1^*= \gamma(\rho_1)-
\gamma(\frac{g_{e1}\rho_1}{1+g_{e2}\rho_2})$
is achieved by setting $\alpha=1$ in the above region.
If $R_K>0$, then the rate $R_1(\alpha)$ is achievable.
As $R_1^* = \max\limits_{\alpha\in [0,1]} R_1(\alpha)$
for $R_K>0$, the point $[R_1^*,0]$ is achievable.
The achievability of $[0,R_2^*]$ can be obtained similarly, and hence,
the region of Theorem~\ref{thm:fullduplexdmc} includes
that of~\cite{HeAndYener}.


\section{Sketch of the Proof of Corollary~\ref{thm:HalfDuplexRandomScheduling}}
\label{prf:HalfDuplexRandomScheduling}

The channel $P^*(y_1,y_2,z|x_1,x_2)$ with states $4$ given to users
reduces to the following equivalent channel.
\begin{eqnarray*}
P^{**}(y_1,y_2,z|x_1,x_2) = \left\{
\begin{array}{c}
P(y_2,z|x_1,x_2=?) {\bf 1}_{\{y_1,?\}}, \textrm{ for state } 1 \\
P(y_1,z|x_1=?,x_2) {\bf 1}_{\{y_2,?\}}, \textrm{ for state } 2 \\
P(z|x_1,x_2) {\bf 1}_{\{y_1,?\}} {\bf 1}_{\{y_2,?\}}, \textrm{ for state } 3 \\
{\bf 1}_{\{y_1,?\}} {\bf 1}_{\{y_2,?\}} {\bf 1}_{\{z,?\}}, \textrm{ for state } 4,
\end{array}
\right.
\end{eqnarray*}

Note that $P^{**}(y_1,y_2,z|x_1,x_2)$ is not equivalent to
$P^*(y_1,y_2,z|x_1,x_2)$. We describe coding scheme for
the channel $P^{**}$. The channel $P^{**}$
will be equivalent to $P^*$, if the nodes can classify the state
$4$ of the channel.

We first consider the channel between $x_1$ and $y_2$
over a block of $n'$ channel uses.
There are $P_1(1-P_2)n'$ symbols for which the channel is
in state 1 (law of large numbers). The symbols for state 2 have
$y_2=?$ are deleted. (These correspond to
symbols that have $x_1=?$.)
The symbols corresponding to state $3$ of the
channel can be modeled as random erasures. (There are
$P_1P_2n'$ such symbols with high probability as $n'$
gets large.) Finally, the channel outputs
corresponding to state 4 will be erased (as there
is no transmission from user 1). Therefore we consider
coding over $[P_1(1-P_2)+P_1P_2]n'$ symbols between
$x_1$ and $y_2$, for which $P_1P_2n'$ symbols are erasures
(as $n'$ gets large).

We first define the followings.
\begin{eqnarray*}
n_1&=&P_1(1-P_2)n'\\
n_2&=&(1-P_1)P_2n'\\
n_3&=&P_1P_2n'\\
n_4&=&(1-P_1)(1-P_2)n'
\end{eqnarray*}

In the codebook design, we generate
$2^{n' (R_1^k+R_1^s+R_1^o+R_1^x)}$ codewords
denoted by $\cv_1^{n_1+n_3}$ of length $n_1+n_3$.
For each symbol time, with probability $(1-P_1)$
we input $x_1=?$ (no transmission event), and with probability
$P_1$ we generate the channel input $x_1$
according to $P(x_1|c_1)$ using the next symbol in
$\cv_1^{n_1+n_3}$. If there is no remaining symbols in $\cv_1^{n_1+n_3}$,
we input $x_1=?$ (the effect of this diminishes
as $n'$ gets large).
Similarly, we generate
$2^{n' (R_2^k+R_2^s+R_2^o+R_2^x)}$ codewords
denoted by $\cv_2^{n_2+n_3}$ of length $n_2+n_3$, and map it to $\xv_2^{n'}$.

For the decodability, the typical set decoding is employed.
For example, the decoder $2$ will select $(w_1^k,w_1^s,w_1^o,w_1^x)$
such that,
\begin{equation*}
(\qv^{n'},\cv_1^{n_1+n_3}(w_1^k,w_1^s,w_1^o,w_1^x),
\yv_2^{n_1+n_3})\in A_{1,\epsilon}^{n_1+n_3}(\textrm{state 1}).
\end{equation*}
Here, the remaining symbols in $\yv_2^{n'}$ are deleted as they are equal to
$?$. The equivalent channel is the random mapping of $\cv_1^{n_1+n_3}$
to $\xv_1^{n_1+n_3}$, from which $n_3$ symbols are randomly erased and the remaining
ones generate $\yv_2^{n_1}$.

Here the error probability (averaged over the ensemble) can be made
small, if
\begin{eqnarray}
R_1^k+R_1^s+R_1^o+R_1^x &\leq& \frac{n_1}{n'}I(C_1;Y_2|X_2,Q,\textrm{state 1}) \label{eq:Appendix2Eq1}\\
R_2^k+R_2^s+R_2^o+R_2^x &\leq& \frac{n_2}{n'}I(C_2;Y_1|X_1,Q,\textrm{state 2})
\label{eq:Appendix2Eq2}
\end{eqnarray}

To show that the secrecy constraint is satisfied, we
follow the steps similar to that of
Appendix~\ref{prf:FullDuplexDMC-R}.
Due to key sharing it suffices to show
\begin{eqnarray*}
\frac{1}{n'}H(W_1^k,W_1^s,W_2^k,W_2^s|\Zm^{n'})&\geq&
R_1^k+R_1^s+R_2^k+R_2^s - \epsilon,
\end{eqnarray*}
for sufficiently large $n'$, together with
\begin{eqnarray}
R_1^o &\leq& R_2^k \label{eq:Appendix2Eq3}, \textrm{ and}\\
R_2^o &\leq& R_1^k \label{eq:Appendix2Eq4}.
\end{eqnarray}

Here, the latter is used to ensure that there are
sufficient number of key bits (from the previous block) to secure
messages that are carried in the open part (of the current block),
and the former is satisfied (from the equivocation
computation provided in Appendix~\ref{prf:FullDuplexDMC-R})
if the rates satisfy the followings.
\begin{eqnarray}
R_1^o+R_1^x &\leq& \frac{n_1+n_2}{n'}I(C_1;Z|C_2,\textrm{state 1 or 2})
+ \frac{n_3}{n'}I(C_1;Z|C_2,\textrm{state 3}) \label{eq:Appendix2Eq5}\\
R_2^o+R_2^x &\leq& \frac{n_1+n_2}{n'}I(C_2;Z|C_1,\textrm{state 1 or 2})
+ \frac{n_3}{n'}I(C_2;Z|C_1,\textrm{state 3}) \label{eq:Appendix2Eq6}\\
R_1^o+R_1^x+R_2^o+R_2^x &=& \frac{n_1+n_2}{n'}I(C_1,C_2;Z|\textrm{state 1 or 2})
+ \frac{n_3}{n'}I(C_1,C_2;Z|\textrm{state 3}) \label{eq:Appendix2Eq7},
\end{eqnarray}

Then the region obtained by equations \eqref{eq:Appendix2Eq1},
\eqref{eq:Appendix2Eq2}, \eqref{eq:Appendix2Eq3}, \eqref{eq:Appendix2Eq4},
\eqref{eq:Appendix2Eq5}, \eqref{eq:Appendix2Eq6}, and
\eqref{eq:Appendix2Eq7} can be simplified (using the same steps given
in Appendix~\ref{prf:fullduplexdmc}) to obtain the stated result.


\section{Proof of Proposition~\ref{prop:halfduplexmodulo}}
\label{prf:halfduplexmodulo}

The proof follows by Corollary~\ref{thm:HalfDuplexRandomScheduling},
where we set $|\Qc|=1$ and compute the followings.
\begin{eqnarray*}
I(C_1;Y_2|X_2,Q,\textrm{state 1})&=&H(\hat{\mu}_2)-H(\hat{\epsilon}_2)\\
I(C_2;Y_1|X_1,Q,\textrm{state 2})&=&H(\hat{\mu}_1)-H(\hat{\epsilon}_1)
\end{eqnarray*}
and the eavesdropper's observed information is given by,
\begin{eqnarray*}
I(C_1,C_2;Z|\textrm{state 3})&=&H(\hat{\mu}_e)-H(\hat{\epsilon}_e)\\
I(C_1,C_2;Z|\textrm{state 1 or 2})&=&\big(H(\mu_{e1}d_1+\mu_{e2}d_2)-0.5H(d_1\epsilon_{e1}+d_2\epsilon_{e2})
- 0.5H(d_1(1-\epsilon_{e1})+d_2\epsilon_{e2})\big),\notag
\end{eqnarray*}
where the last equality is a direct results of the following computation.
\begin{eqnarray*}
&&H(Z|C_1=0,C_2=0)=H(d_1\epsilon_{e1}+d_2\epsilon_{e2})
\\&&H(Z|C_1=1,C_2=1) = H(Z|C_1=0,C_2=0)
\\&&H(Z|C_1=1,C_2=0)=H(d_1(1-\epsilon_{e1})+d_2\epsilon_{e2})
\\&&H(Z|C_1=0,C_2=1)=H(Z|C_1=1,C_2=0)
\end{eqnarray*}


\section{Proof of Theorem~\ref{thm:tdm}}
\label{prf:tdm}

Consider the time intervals when Alice is transmitting codewords to Bob.
Let $\alpha_{M}$, $\alpha_{E}$ denote the fraction of
symbols erased at Bob and Eve, and ${P_{e}}^{(M)}$,
$P_{e}^{(E)}$ denote the probability of erroneously
decoding a received symbol given that it was not erased
at Bob and Eve, respectively. By applying the appropriate
random binning scheme~\cite{wyner}, the following secrecy
rate is achievable (\cite{BCC}, Theorem 3).
\begin{eqnarray*}
R = \max\limits_{P(x)} \left\{[I(X;Y)-I(X;Z)]^{+}\right\},
\end{eqnarray*}
where $X$ denotes the input, $Y$ and $Z$ denote the outputs
at Bob and Eve, respectively. Considering the transition model
for this channel, we see
\begin{eqnarray*}
H(Y|X) = H(\alpha_{M}) + (1 - \alpha_{M})H({P_{e}}^{(M)}).
\end{eqnarray*}

Now, let $\textrm{Pr}\{X(t)=\sqrt{\rho(t)}\}=\Pi$
and $\textrm{Pr}\{X(t)=-\sqrt{\rho(t)}\}=1-\Pi$. Then,
\begin{eqnarray*}
H(Y) = H(\alpha_{M}) + (1 - \alpha_{M})H(\Pi (1 - {P_{e}}^{(M)})
+ (1 - \Pi) {P_{e}}^{(M)}),
\end{eqnarray*}
and $\max_{\Pi} H(Y) = H(\alpha_{M}) + (1 - \alpha_{M})$
when $\Pi = 0.5$. This results in
\begin{eqnarray*}
\max\limits_{P(x)} I(X;Y) \:=\:
\max\limits_{P(x)} (H(Y)-H(Y|X)) \:=\: (1-\alpha_{M})(1-H({P_{e}}^{(M)}))
\end{eqnarray*}
Similarly, $\max\limits_{P(x)} I(X;Z) = (1-\alpha_{E})(1-H(P_{e}^{(E)}))$.

Following the half-duplex assumption, all data symbols transmitted
during the same time interval of a feedback transmission will be
considered as erasures at the legitimate receiver's channel.
Therefore, as the frame length $T\rightarrow\infty$,
$\alpha_{M} = \beta$. For the rest of the symbols, the
probability of symbol error by the hard decision detector will be
\begin{eqnarray*}
{P_{e}}^{(M)}(t) = 1-\phi\left(\sqrt \frac{\rho(t)}{{d_{AB}}^{\alpha}}\right).
\end{eqnarray*}

On the other hand, feedback transmissions will introduce decoding
errors at Eve. Noting that $1 - P_{m}$ of those
corrupted symbols will be detected by the energy classifier,
we get
\begin{eqnarray*}
\alpha_{E} = \beta (1 - P_{m}) + (1 - \beta) P_{f}
\end{eqnarray*}
\begin{eqnarray*}
P_{e}^{(E)} = \frac{\beta P_{m}P_{e|m}}{1-\alpha_{E}}.
\end{eqnarray*}

Combining these results, we obtain
\begin{eqnarray*}
\max\limits_{P(x)}I(X;Y)
&=& (1 - \beta)\left(1 - H\left(\frac{1}{T}\displaystyle\sum_{t=1}^{T}{P_{e}}^{(M)}(t)\right)\right)\\
&\geq& (1 - \beta)\left(1 - H\left(1-\phi\left(\sqrt \frac{\rho_{\textrm{min}}}{{d_{AB}}^{\alpha}}\right)\right)\right)\\
&\triangleq& R_{M}
\end{eqnarray*}
and denoting $R_{E}\triangleq (1 - \alpha_{E})(1 - H(P_{e}^{(E)}))$,
we have $\max\limits_{P(x)}I(X;Z)=R_{E}$, and
\begin{eqnarray*}
R \:=\: \max_{P(x)} ([I(X;Y)-I(X;Z)]^{+})
\:\geq\: [\max_{P(x)}I(X;Y)-\max_{P(x)}I(X;Z)]^{+}
\:\geq\: [R_{M}-R_{E}]^{+}.
\end{eqnarray*}
Finally, we consider a {\em max-min} strategy whereby
the legitimate receiver assumes that the eavesdropper chooses
its position around the perimeter of the circle and the energy
classifier's mechanism ${\cal C}$ to minimize the secrecy rate
$R_{s}$. Accordingly, the legitimate receiver determines the
probability of random feedback transmission $\beta$ and both
the data and feedback signal power distributions $f_{1}$ and $f_{2}$
to maximize this worst case value (note that the rate is scaled
by $0.5$ to account for the time division between the two nodes).
We obtain
\begin{eqnarray*}
R_{s} = 0.5\max_{\beta,f_{1},f_{2}}\left\{\min_{\theta,\cal{C}}R\right\}.
\end{eqnarray*}


\section{Proof of Theorem~\ref{thm:interactive}}
\label{prf:interactive}

Due to symmetry, we only consider the secrecy rate of
Alice's message to Bob. Following the previous proof, we have the following achievable secrecy rate,
\begin{eqnarray*}
R &=& [(1-\alpha_{M})(1-H({P_{e}}^{(M)}))
-(1-\alpha_{E})(1-H({P_{e}}^{(E)}))]^{+},
\end{eqnarray*}
where $\alpha_{M}$, $\alpha_{E}$ denote the fraction of
symbols erased at Bob and Eve, and ${P_{e}}^{(M)}$,
$P_{e}^{(E)}$ denote the probability of erroneously
decoding a received symbol given that it was not erased
at Bob and Eve, respectively.
Using half-duplex antennas, each node will be able to
decode a symbol transmitted by the other node only when
its own transmitter is idle and the other node's transmitter
is active. These two conditions are simultaneously satisfied
with probability $P_{t}(1 - P_{t})$ yielding $\alpha_{M} =
1 - P_{t}(1 - P_{t})$. We also see that
\begin{eqnarray*}
{P_{e}}^{(M)}(t)=1 - \phi\left(\sqrt \frac{\rho(t)}{{d_{AB}}^{\alpha}}\right).
\end{eqnarray*}

The symbols classified by Eve as being transmitted by Alice can
belong to one of three categories. The first, which takes place
with probability $P_{t}\left(1 - P_{t}\right)\left(1 -
P_{(A,B^{c})\rightarrow(A^{c},B)} - P_{(A,B^{c})\rightarrow(A,B)}\right)$,
represents the portion successfully detected and correctly decoded by
Eve. The second corresponds to symbols transmitted by Bob and
misclassified as belonging to Alice;
with probability $P_{t}\left(1 - P_{t}\right)P_{(A^{c},B)\rightarrow(A,B^{c})}$.
Those symbols are independent from the ones transmitted by Alice, and hence,
have a probability $0.5$ of being different. The third category, with
probability $P_{t}^{2}P_{(A,B)\rightarrow(A,B^{c})}$, corresponds to
concurrent transmissions that are not {\em erased} by Eve's classifier
and misclassified as Alice's symbols. The probability of error in these
symbols is denoted by $P_{e|(A,B)\rightarrow(A,B^{c})}$.
Combining these, we get
\begin{eqnarray*}
\alpha_{E} = 1 - D_{A}
\end{eqnarray*}
\begin{eqnarray*}
P_{e}^{(E)} = \frac{P_{e}^{(EA)}}{1-\alpha_{E}}
\end{eqnarray*}
\begin{eqnarray*}
R &=& \left[(1-\alpha_{M})(1-H({P_{e}}^{(M)}))-(1-\alpha_{E})
\left(1-H\left(P_{e}^{(E)}\right)\right)\right]^{+}
\\&\geq&\left[P_{t}(1-P_{t})\left(1-H\left(1 - \phi\left(\sqrt \frac{\rho_{\textrm{min}}}{{d_{AB}}^{\alpha}}\right)\right)\right)
-D_{A}\left(1-H\left(\frac{P_{e}^{(EA)}}{D_A}\right)\right)\right]^{+}
\end{eqnarray*}
And the same result applies to the secrecy rate of Bob's message to Eve by using,
\begin{eqnarray*}
\alpha_{E} = 1 - D_{B}
\end{eqnarray*}
\begin{eqnarray*}
P_{e}^{(E)} = \frac{P_{e}^{(EB)}}{1-\alpha_{E}}
\end{eqnarray*}
Finally, in order to achieve symmetric secure communication, we set both rates to the minimum of achievable secrecy rates for the two nodes. We follow the same min-max strategy as given in the proof of Theorem~\ref{thm:tdm}
to obtain the lower bound on $R_s$.


\section{Proof of Corollary~\ref{cor:maxinteractive}}
\label{prf:maxinteractive}

By ignoring the noise effect at Eve, symbols where both transmitters are active will be correctly decoded at Eve as the symbol with the highest transmit power. Hence, with no prior information regarding the source of any transmitted symbol, Eve will not erase any symbol, i.e. $E_2 \in \left\{(A,B^{c}),(A^{c},B)\right\}$. Moreover, $P_{E_1 \rightarrow E_2}=0.5$ for all six possible combinations of $E_1$ and $E_2$, $P_{e|(A,B) \rightarrow E_2}$=0.25 for the two possible values of $E_2$. By applying those values, we get:
\begin{equation*}
R_{EA}=R_{EB}= P_{t}(1-0.5P_{t})(1-H(0.25))
\end{equation*}
These values are achieved by employing a symmetric \emph{real-time} detector at Eve, i.e. $R_{EA}=R_{EB}$, and each symbol has to be decoded as being transmitted either by Alice or Bob. However, Eve may choose to maximize the value $\max(R_{EA},R_{EB})$ by either maximizing only one of those values at the cost of minimizing the other, or by allowing its decoder to \emph{match} the same symbol to different sources, e.g., let $P_{E_1 \rightarrow (A,B^{c})}= 1$ for all possible values of $E_1$, then,
\begin{eqnarray*}
D_{A} = 1-(1-P_{t})^{2},
\end{eqnarray*}
note that $P_{e|(E_1 \rightarrow (A,B^{c}))}$ remains the same. By applying the resulting probabilities in the last example, we get the rate in (\ref{eqn:upperbound}). It is obvious that by symmetry, having $E_2=(A^{c},B)$ for all symbols results in the same rate.


\bibliographystyle{IEEEtran}


\begin{figure}[htb]
\centering
\includegraphics[width=0.48\columnwidth]{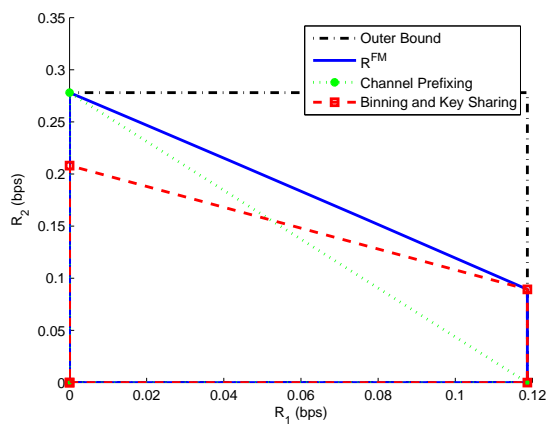}
\caption{Boundaries of achievable rate regions for the modulo-$2$
channel, when $\epsilon_1=0.2$, $\epsilon_2=0.3$, $\epsilon_e=0.25$,
and $\mu_1=\mu_2=0.5$. The outer bound is the capacity of
the two-way channel without the secrecy constraints.}
\label{fig:modulo}
\end{figure}

\begin{figure}[htb]
\centering
\includegraphics[width=0.48\columnwidth]{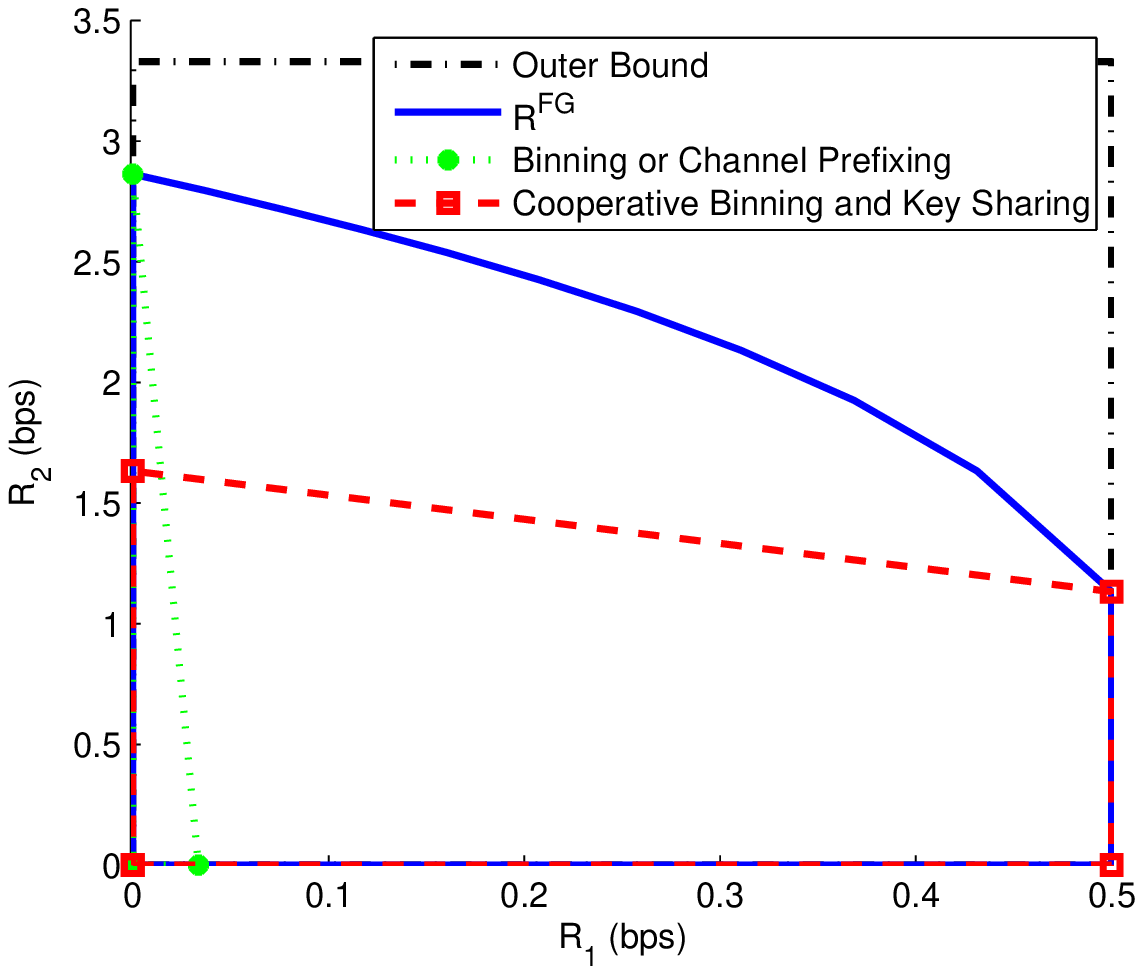}
\caption{Boundaries of achievable rate regions for the Gaussian
channel, when $g_{11}=g_{22}=1$, $g_{e1}=10$, $g_{e2}=0.1$, and
$\rho_1=1$, $\rho_2=100$. The outer bound is the capacity of
the two-way channel without the secrecy constraints. }
\label{fig:gaussian}
\end{figure}

\begin{figure}[t]
\centering
\includegraphics[width=0.48\columnwidth]{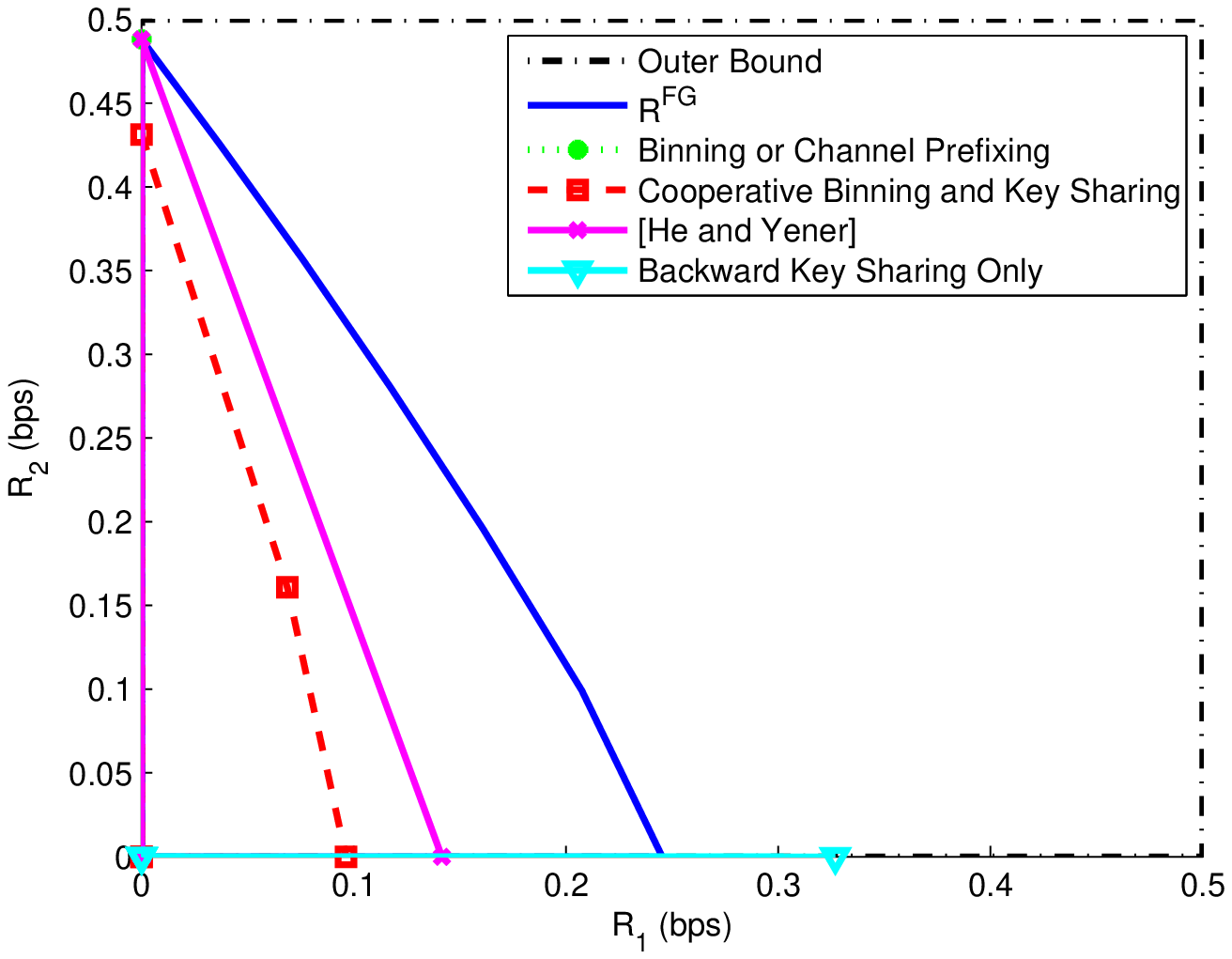}
\caption{Boundaries of achievable rate regions for the Gaussian
channel, when $g_{11}=g_{22}=1$, $g_{e1}=5$, $g_{e2}=0.1$, and
$\rho_1=\rho_2=1$. The outer bound is the capacity of
the two-way channel without the secrecy constraints.}
\label{fig:gaussian2}
\end{figure}

\begin{figure}[htb]
\centering
\includegraphics[width=0.24\columnwidth]{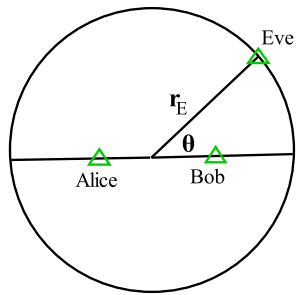}
\caption{Near field wireless communications scenario. Eve is assumed to be located outside a circle of radius $r_E$ whose center lies at the mid-point between Alice and Bob.}
\label{fig:BANmodel}
\end{figure}

\begin{figure}[htb]
\centering
\includegraphics[width=0.48\columnwidth]{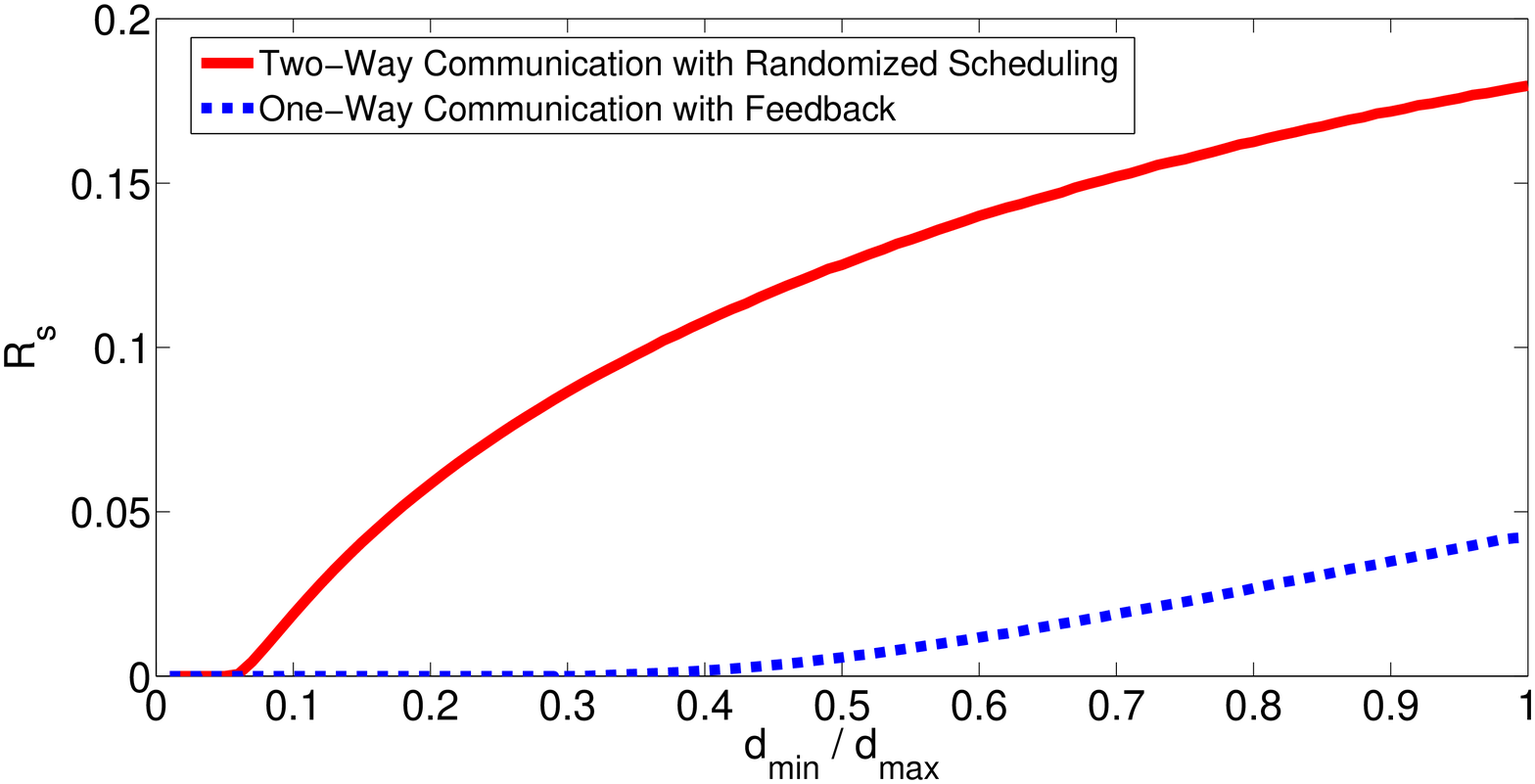}
\caption{Maximum achievable secrecy rate for different distance ratios between Eve and each of the two communicating nodes.}
\label{fig:numerical}
\end{figure}

\begin{figure}[htb]
\centering
\includegraphics[width=0.48\columnwidth]{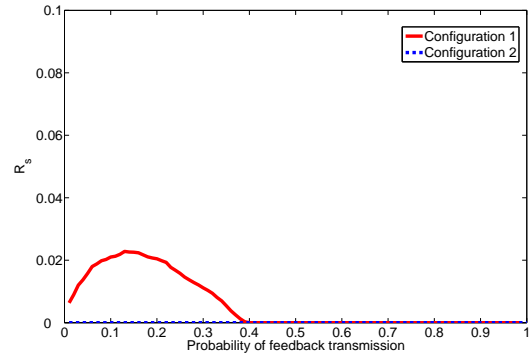}
\caption{$\beta$ vs. $R_s$ in different configurations for the one way TDM scheme, $R_s$ = 0.5$[R_{M}-R_{E}]^{+}$. We consider the case when Alice is the transmitter and Bob is the legitimate receiver.}
\label{fig:tdm}
\end{figure}

\begin{figure}[htb]
\centering
\includegraphics[width=0.48\columnwidth]{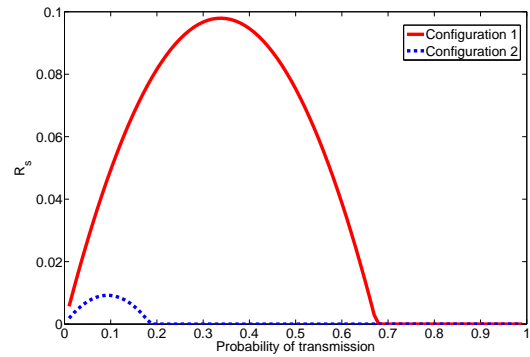}
\caption{$P_{t}$ vs. $R_s$ in different configurations for the randomized scheduling communication scheme, $R_s$ = [$R_{M}$-$\max$($R_{EA}$,$R_{EB}$)]$^{+}$.}
\label{fig:interactive}
\end{figure}

\end{document}